\documentclass[aps,pra,longbibliography,notitlepage,superscriptaddress,11pt,tightenlines]{revtex4-1}
\usepackage{amsmath, amssymb, amsfonts, graphicx, bm, bbm, times, xcolor, amsthm, subfigure, array, multirow, booktabs}
\usepackage[T1]{fontenc}
\usepackage{hyperref}
\definecolor{darkblue}{RGB}{0,0,128} 
\definecolor{darkgreen}{RGB}{0,150,0}
\hypersetup{breaklinks, colorlinks, linkcolor=darkblue, citecolor=darkgreen, filecolor=red, urlcolor=darkblue, 
	pdftitle={Performance of quantum error correction with coherent errors}, 
	pdfauthor={Eric Huang, Andrew C. Doherty, Steven Flammia}}
\usepackage{cleveref}

\graphicspath{{./figures/}}

\newcolumntype{C}{>{$}c<{$}}
\AtBeginDocument{
\heavyrulewidth=.08em
\lightrulewidth=.05em
\cmidrulewidth=.03em
\belowrulesep=.65ex
\belowbottomsep=0pt
\aboverulesep=.4ex
\abovetopsep=0pt
\cmidrulesep=\doublerulesep
\cmidrulekern=.5em
\defaultaddspace=.5em
\tabcolsep=7pt
}

\newtheorem{theorem}{Theorem}

\newcommand{\cE}{\mathcal{E}}

\renewcommand{\vec}[1]{{\bm{#1}}}

\newcommand{\acdedit}[1]{#1}

\begin{document}

\title{Performance of quantum error correction with coherent errors}
\author{Eric Huang}
\author{Andrew C.\ Doherty}
\affiliation{Centre for Engineered Quantum Systems, School of Physics, The University of Sydney, Sydney, Australia}
\author{Steven Flammia}
\affiliation{Centre for Engineered Quantum Systems, School of Physics, The University of Sydney, Sydney, Australia}
\affiliation{Yale Quantum Institute, Yale University, New Haven, CT 06520, USA}
\date{\today}

\begin{abstract}
We compare the performance of quantum error correcting codes when memory errors are unitary with the more familiar case of dephasing noise. 
For a wide range of codes we analytically compute the effective logical channel that results when the error correction steps are performed noiselessly. 
Our examples include the entire family of repetition codes, the 5-qubit, Steane, Shor, and surface codes. 
When errors are measured in terms of the diamond norm, we find that the error correction is typically much more effective for unitary errors than for dephasing. 
We observe this behavior for a wide range of codes after a single level of encoding, and in the thresholds of concatenated codes using hard decoders. 
We show that this holds with great generality by proving a bound on the performance of any stabilizer code when the noise at the physical level is unitary. 
By comparing the diamond norm error $D'_\diamond$ of the logical qubit with the same quantity at the physical level $D_\diamond$, we show that $D'_\diamond \le c D^d_\diamond $ where $d$ is the distance of the code and $c$ is \acdedit{a} constant that depends on the code but not on the error. 
This bound compares very favorably to the performance of error correction for dephasing noise and other Pauli channels, where an error correcting code of odd distance $d$ will exhibit a scaling $D'_\diamond \sim D_\diamond^{(d+1)/2}$. 
\end{abstract}

\pacs{} 

\maketitle

\section{Introduction}

Building a large-scale quantum computer will require substantial efforts to mitigate noise through the use of quantum error correction and fault tolerance. 
The fault tolerance threshold theorem~\cite{Knill1996, Aharonov1997, Kitaev1997a, Shor1996} guarantees that as long as the errors are sufficiently rare and weakly correlated, an arbitrarily long quantum computation can proceed indefinitely and with low overhead.
The exact numerical value of the threshold depends critically on the assumptions about the noise, and from the perspective of fault-tolerant quantum computing not all types of errors are equivalent. 
This is even true for uncorrelated noise, since errors such as dephasing and depolarizing noise are purely stochastic, but control errors such as unitary over- or under-rotation can add coherently. 

This distinction between stochastic and coherent errors was recognized quite early on to be important~\cite{Kitaev2002}. 
In particular, using our best theorems to date, the only known way to relate the threshold for stochastic errors to the threshold for coherent errors is to \emph{square} the stochastic value of the threshold~\cite{Aliferis2006a}. 
Thresholds quoted in the literature for stochastic-type noise range between values of about $10^{-2}$ to $10^{-4}$ depending on how generous the assumptions are on the stochastic noise and whether the threshold is a numerical estimate based on simulation~\cite{Wang2011} or a theorem based on rigorous proof~\cite{Aliferis2007} (or something intermediate). 
Not knowing if these values need to be squared therefore represents a rather large gap in our understanding of the threshold. 
The situation is complicated by the fact that the only large-scale numerical simulations that are tractable must necessarily deal with Pauli errors, for which it is known that the squaring is unnecessary~\cite{Magesan2012}. 

The role that non-Pauli errors play in the fault-tolerance threshold is therefore quite poorly understood. 
Some recent works are beginning to develop our understanding, such as the use of the so-called honest Pauli approximation~\cite{Puzzuoli2014} or recent full-scale simulations of coherent noise, using small codes~\cite{Gutierrez2016}, using tools such as tensor networks~\cite{Darmawan2017}, or in some special cases via an exact solution~\cite{Greenbaum2016, Suzuki2017, Bravyi2018}. 

The focus of this paper is on understanding the role of coherent errors in quantum error correction. 
In order to motivate our main results, it is helpful to focus the discussion on two meaningful quantities that we wish to study, the average gate infidelity $r$ and the diamond distance $D_\diamond$, defined below.
This will also help motivate the particular scaling behavior that we seek to quantify. 

We would like to compare quantities that can actually be measured in experiments directly to the fault-tolerance threshold. 
Unfortunately, here we see another large gap between what we can measure and what we can infer about the threshold. 
For example, the average gate infidelity $r(\cE)$, defined as 
\begin{align}
\label{eq:ravgdef}
	r(\cE) = 1-\int \mathrm{d}\psi \langle \psi | \cE(\psi) | \psi \rangle,
\end{align}
is routinely measured to high precision in randomized benchmarking experiments~\cite{Knill2008}. 
However, all of our provable knowledge about thresholds uses a much stricter error metric, the diamond distance $D_\diamond(\cE)$ to the identity, defined as 
\begin{align}
\label{eq:diamonddef}
	D_\diamond(\cE) = \frac{1}{2} \| \cE - \mathcal{I} \|_\diamond := \frac{1}{2} \sup_{\rho} \bigl\| \cE\otimes\mathcal{I}(\rho) - \rho\bigr\|_1
\end{align}
in terms of the Schatten 1-norm (the sum of the singular values).
The gap between these two quantities in the regime of interest can be orders of magnitude in general, a point recently emphasized by \citet{Sanders2015}. 
Unfortunately, there is no simple way to measure the diamond distance in general without doing complete process tomography~\cite{Blume-Kohout2016}. 

Recent work has sought ways to quantify the worst-case behavior needed to prove a threshold theorem using measurement methods that are preferably scalable and avoid tomography. 
Examples include gate set tomography~\cite{Blume-Kohout2016} and the unitarity~\cite{Wallman2015}, though neither method is completely scalable in contrast to randomized benchmarking. 

Complementing this line of research, and motivating the upper bound we prove below in Theorem 1, is the theoretical approach of~\cite{Kueng2016, Wallman2015b}, which seeks to classify physical noise sources in terms of ``good'' and ``bad'' noise scaling. 
A family $\{\cE_\gamma:\gamma\in[0,1]\}$ of noise models with $\cE_0=\mathcal{I}$ has ``good'' scaling if $D_\diamond(\cE_\gamma)\leq Cr(\cE_\gamma)$ and ``bad'' scaling if $D_\diamond(\cE_\gamma) \geq C'\sqrt{r(\cE_\gamma)}$ for all $\gamma\in[0,1]$ and some constants $C,C'> 0$ that are independent of $\gamma$. 
Other scalings are also possible, and could be called ``intermediate''. 
The purpose of this coarse distinction is that for a given noise channel $\cE$, if $D_\diamond(\cE) \approx r(\cE)$ then the proxy measure of $r(\cE)$ that is easily obtained via benchmarking gives a good indication of how close one might be to the threshold, while if $D_\diamond(\cE) \approx \sqrt{r(\cE)}$ then the proxy is highly misleading. 
These scaling limits are extremal~\cite{Wallman2014, Kueng2016, Sanders2015, Wallman2015b}. 
Of course, the usefulness of this distinction depends implicitly on the constants $C$ and $C'$ being relatively civilized. 
Taking, say, $C=10^{15}$ shows that all noise is trivially good until $r(\cE_\gamma)\lesssim 10^{-30}$. 
Similarly, we also need $C'$ to be sufficiently large for this distinction to be meaningful in practice. 
However, while we are always ultimately interested in absolute noise rates on single instances, hiding these constants and discussing scaling enables us to make important physical insights into the nature of noise sources and what the expected effect might be on 
quantum information. 

Illustrating the utility of this scaling behavior perspective, Ref.~\cite{Kueng2016} classifies many common qubit noise models such as dephasing, depolarizing, amplitude damping, leakage, and unitary errors according to this dichotomy. 
It is only the unitary errors and so-called ``coherent leakage'' that exhibit bad error scaling~\cite{Kueng2016}. 
By combining some knowledge of the dominant noise process with measurements of $r(\cE)$ and the aforementioned unitarity, Refs.~\cite{Kueng2016, Wallman2015b} show that it is possible to obtain pertinent information about how close one's qubits are to the fault-tolerance threshold. 

This discussion of scaling focused on the case of \emph{physical} errors, but even more desirable would be to understand the scaling of \emph{logical} error rates in quantum error correcting codes~\cite{Combes2018}. 
Motivated by the above notions of good and bad error scaling in regards to physical errors, we seek to develop an analogous understanding of how coherent and incoherent errors scale at the logical level. 
It has been observed numerically that there can be orders of magnitude difference in the logical error rates after concatenating quantum codes~\cite{Iyer2017}, so \textit{a priori} it is not clear that it is possible to obtain a simple scaling dichotomy as we seem to have in the case of physical errors. 

\subsection{Summary of Results}

In this paper, we focus on physical qubit noise channels with both unitary control errors and dephasing noise. 
We aim to compare the performance of ideal error correction for coherent noise channels, for which the diamond norm distance greatly exceeds the average fidelity error, with dephasing, for which they are comparable. 
The main tool is to calculate effective noise channels for the encoded qubits after error correction, following~\citet{Rahn2002}. 
\acdedit{Broadly, the results of our investigations into the performance of ideal error correction under coherent errors suggest} that simply considering diamond norm error at the physical level overstates the effect that residual coherent errors will have on logical information in a quantum computer.

In order to develop our analytical methods we begin in \cref{sec:repetition} by revisiting the case of the repetition codes studied by \citet{Greenbaum2016}. 
Although these are not proper quantum error correcting codes they can correct against $Z$-rotations and dephasing, and are simple enough to be analytically tractable for arbitrarily high distance codes. 
As has been shown previously~\cite{Iyer2017}, the performance of error correction depends very strongly on how coherent the noise process is, and is not predicted by a single figure of merit for the noise, such as average fidelity or diamond norm distance. 

One apparently new result of this analysis relates to the nature of the effective noise channels for codes where the distance is odd or even. 
If the noise process is a purely unitary rotation about $Z$ and the repetition code has odd distance we find that the effective noise channel conditioned on syndrome measurement is also unitary. 
Consequently the statistics of syndrome measurement are independent of the logical state of the code. 
For the case of even distance repetition codes, however, we find that the effective channel conditioned on the syndrome measurement corresponds to a (very) weak measurement of logical $Z$. 
As a result the syndrome measurement outcomes do depend weakly on the logical state of the qubit. 
However, the effective noise channel averaged over syndromes is pure dephasing, even when the underlying physical noise process is purely unitary. 
We have observed qualitatively similar behavior for even distance surface codes.

In \cref{sec:theoremsection} we apply these methods of calculation of effective channels to general stabilizer codes. 
Our main analytical result is \cref{thm:theorem}, which roughly states the following. 
For any $[n,k,d]$ stabilizer code, the logical diamond norm error $D'_\diamond$ after ideal error correction of a purely unitary error is bounded by $D'_\diamond \leq c_{n,k} D^d_\diamond$ in terms of the physical diamond norm error $D_\diamond$, where $c_{n,k}$ is a constant independent of the errors (but may depend on the code). 
This result does not just apply to uniform unitary noise, but is readily generalized to hold for all single qubit unitary noise if we express the bound in terms of the largest value of $D_\diamond$ across the physical qubits.
This is very favorable behavior compared to the performance of error correction for dephasing noise. 
Recall that an error correcting code of odd distance $d$ will correct $t$ dephasing errors where $d=2t+1$. 
Consequently for dephasing it is known that $D'_\diamond$ scales like $D_\diamond^{t+1}$. 
Another way to assess this comparison between coherent and dephasing noise after error correction is to express this bound in terms of the average fidelity error $r$ at the physical level $D'_\diamond \le c r^{t+1/2}$, where $c$ is a constant that depends on the code but not on the unitary error. 
(Recall that for a coherent error $ D_\diamond$ scales like  $r^{1/2}$.) 
This compares quite favorably to the well-known scaling $r^{t+1}$ for dephasing errors. 

In \cref{sec:unitarycodeperformance} we compute the effective noise channels for a range of quantum codes using an automated procedure. 
We consider unitary qubit noise and study how the performance of the code depends on the axis of rotation in the Bloch sphere. 
There are frequently large effects. 
For example, for the surface code the error correction performs much better for rotations about the Y axis, than for rotations about X and Z. 
This observation is analogous to the recently discussed behavior of the surface code under Pauli noise~\cite{Tuckett2018}.

Finally, in \cref{sec:concatcodeperformance} we study concatenated codes for noise channels that combine $Z$-rotations with dephasing using a hard decoder. 
Again we use an automated procedure to generate the effective channel for a single level of encoding. 
By regarding these effective channels as a map on noise processes it is possible to find the threshold for a hard decoder as a fixed point of the map. 
We compare the two extreme limits of unitary $Z$ rotations and pure dephasing and find that the threshold as a function of the diamond norm error is in every code we tested larger for the case of unitary errors. 

\section{Noise models and effective channels}

Suppose we have a qubit Hamiltonian $H$ with $H^2=I$, then coherent noise channels result from unitary noise processes of the form
\begin{align}
    U_{H,\theta} = \mathrm{e}^{-i\theta H} = \cos\theta I -i\sin \theta H.
\end{align}
Physically, such error processes arise from over or under rotation in qubit control pulses. 
Note that since we are considering qubit noise processes, the choice $H^2=I$ allows for a general \acdedit{Bloch sphere rotation axis.}  

It is frequently of interest to combine coherent and incoherent errors. 
To study this sort of noise process we will follow~\citet{Kueng2016} and specialize to rotations about the $Z$ axis of angle $\theta$ and dephasing with probability $p$. \acdedit{This simple phenomenological model allows us to make a detailed comparison between the effects of coherent and incoherent errors, but it also corresponds to the dominant noise processes in many experimental implementations of quantum computing.}
The resulting noise channel is 
\begin{align}
    \mathcal{N}(\rho)
    &= \mathrm{e}^{-i\theta Z}\left[(1-p)\rho + pZ\rho Z\right]\mathrm{e}^{i\theta Z}\\
    &= (1-x)\rho + xZ\rho Z - iy(Z\rho - \rho Z). \label{eq:noisemodel}
\end{align}
In the second expression we have used an alternative parameterization of the noise process that will greatly simplify certain calculations. 
The real parameters $x,y$ are defined as follows
\begin{align}
    x
    &= p\cos^2\theta + (1-p)\sin^2\theta \\
    y &= (1-2p)\cos\theta\sin\theta.
\end{align}
It has become common in the literature to compare a noise model such as this with \acdedit{some} Pauli channel that approximates it~\cite{Puzzuoli2014, Gutierrez2013, Gutierrez2015, Gutierrez2016}. 
For this example \acdedit{the Pauli Twirling Approximation to this error} is just to project down to the case where $y=0$. 
Thus for a given initial noise model given by $p,\theta$ the \acdedit{Pauli Twirling Approximation} is to consider instead the model with $p'=p\cos^2\theta+(1-p)\sin^2\theta$ and $\theta'=0$.

Since we are mainly interested in the scaling when error rates are small, we will often expand expressions such as $x$ and $y$ to lowest nontrivial order in $p$ and $\theta$. 
So for example, we have $x \approx p + \theta^2$ and $y \approx \theta$. 

In the following we will also occasionally be interested in the noise operations conditional on the outcome of a syndrome measurement in an error correction procedure. 
These noise processes will likewise combine unitary and dephasing error but will not necessarily preserve trace, being of the form
\begin{align}
    \mathcal{N}_{\vec{s}}(\rho)
       &= \bar{x}\rho + xZ\rho Z - iy(Z\rho - \rho Z). \label{eq:condnoisemodel}
\end{align}
with the case $\bar{x}=1-x$ preserving the trace, and $x+\bar{x}\leq 1$ corresponding to some stochastic process occurring with probability $x+\bar{x}$. 
Here the subscript $\vec{s}$ labels a particular syndrome. 

We will frequently be concerned with measures of the strength of the noise process. 
Though there are many possible choices, we will confine ourselves to two frequently used measures. 
The first is the average gate infidelity $r$ because it can be estimated accurately in randomized benchmarking experiments. 
The second is the diamond norm error $D_\diamond$ because it can be used to bound the overall error when noise processes can occur sequentially in a computation and therefore appears in the statement of fault tolerance threshold theorems. 
For our model channel both of these error metrics were calculated by~\citet{Kueng2016} for the above unitary and dephasing channel: 
\begin{align}
    r &=\frac{2x}{3} =  \frac{2}{3}\bigl[p\cos(2\theta)+\sin^2(\theta)\bigr] \approx \frac{2}{3}\left(p+\theta^2\right)\\
    D_\diamond &= \sqrt{x^2+y^2} = \frac{1}{2}\left|1-(1-2p)\mathrm{e}^{2i\theta}\right| \approx \sqrt{p^2 +\theta^2}.
\end{align}
Notice that in the limit $p \rightarrow 0$ corresponding to a unitary error, $D_\diamond = \lvert\sin\theta\rvert \approx \lvert\theta\rvert$ while $r\approx \theta^2$ so that the diamond norm distance can be much larger than the average fidelity error~\cite{Kueng2016, Sanders2015}. 
Also notice the simplifications in these formulas when expressed in terms of $x$ and $y$, and in particular that $r$ is independent of $y$, while the diamond norm distance is just the length of the vector $(x,y)$. 

We wish to study the performance of $[n,k,d]$ quantum error correcting codes that encode $k$ logical qubits in $n$ physical qubits with code distance $d$. 
As discussed in the introduction, we will focus on ideal error correction where noise processes act only on the qubit memory and the encoding and error correction are performed ideally. 
Such an ideal error correction process defines an effective channel on the logical qubits. 
We wish to build intuition about the effectiveness of error correction for coherent errors by studying these error channels analytically.
 
The effective channel on the logical qubits arises from composing the operations of encoding, noise, error correction then decoding. 
As such it is described by a completely positive trace-preserving map on the $k$ logical qubits.
It can be written in the form
\begin{align}
    \mathcal{N}_L = \mathcal{E}^\dagger \mathcal{R} \tilde{\mathcal{N}} \mathcal{E}
\end{align}
where $\mathcal{E}$ is the encoding map (an isometry), $\tilde{\mathcal{N}}$ is the noise acting on the $n$ qubits of the code, $\mathcal{R}$ is recovery by syndrome measurement with correction and $\mathcal{E}^\dagger$ reverses the encoding. 
The application of quantum operations in our notation follows the same conventions as matrix multiplication, with operations occurring earlier being written to the right of later operations.

For independent qubit noise on an $n$-qubit code the noise process acts on each physical qubit independently and  
\begin{align}
	\tilde{\mathcal{N}}=\bigotimes_{m=1}^{n}\mathcal{N},
\end{align}
where the single qubit noise process $\mathcal{N}$ is assumed to act independently and identically on each physical qubit. 
Most of our considerations can be generalized to the case of non-identical noise processes acting independently on the different physical qubits, but for now we defer a discussion of this case. 

The error correction process is made up of syndrome measurement followed by correction. 
We consider two natural cases.
First, to obtain a trace-preserving map, we average over all syndromes to obtain the error correction map $\mathcal{R}$. 
Physically this corresponds to the error process assigned by an observer external to the error correction process who is unaware of which syndrome arose. 
Alternatively we can evaluate the conditional operation where a particular syndrome ${\vec{s}}$ is specified and the recovery map $\mathcal{R}_{\vec{s}}$ applies the projection corresponding to this syndrome and then performs the subsequent correction. 
The resulting conditional noise process on the logical qubits $\mathcal{N}_{\vec{s}}$ will be completely positive but not trace preserving; it will be normalized by the probability of the specified syndrome.

Techniques for calculating the effective noise channel $\mathcal{N}_L$ are discussed at length by~\citet{Rahn2002}. 
As described there we can regard the error correction procedure as a map on qubit noise channels $\mathcal{N}\rightarrow \mathcal{N}_L$. 
Note that while the encoding, error correction and decoding operations are linear when regarded as a map on the noise process $\tilde{\mathcal{N}}$, the map $\mathcal{N}\rightarrow \tilde{\mathcal{N}}$ is a polynomial map on the matrix elements that define $\mathcal{N}$. 
Specifically, we can expand a general qubit noise process as follows
\begin{equation}
\mathcal{N}(\rho)=\sum_{P,P'}n_{PP'}P\rho P'
\end{equation}
where the sum runs over all qubit Pauli matrices. 
Then the coefficients in the corresponding expansion for either $\mathcal{N}_L$ or $\mathcal{N}_{\vec{s}}$ are $n$th-order polynomials of the coefficients $n_{PP'}$. 
In the case of the noise process \cref{eq:noisemodel} we conclude that the effective channels after error correction can be written as polynomials of the parameters $x$ and $y$. 
We will make use of this observation to simplify analytical calculations in the following section.

In the majority of our examples below we have used Mathematica to automate the computation of the effective channel $\mathcal{N}_L$ using the approach outlined by~\citet{Rahn2002}. 
The results of those calculations are discussed in detail in \cref{sec:unitarycodeperformance,sec:concatcodeperformance}. 
However, when the error correcting code and recovery map have high symmetry it is possible to compute the effective channel explicitly for unitary error processes, and we will now describe that process using the repetition codes as a specific example (this case was also considered by \citet{Greenbaum2016}). 
The techniques developed to analyze these codes will then be sufficient to allow us to place bounds on the diamond norm error after correction for general $[n,k,d]$ stabilizer codes.

\section{Effective error channels for coherent errors and repetition codes}\label{sec:repetition}

In this section we will describe how to calculate effective error channels for coherent error processes, using the repetition codes as a primary example. 
The high symmetry of these codes allows us to find simple closed form expressions for essentially all quantities of interest. 
Despite the fact that these are not proper quantum error correcting codes, the behavior of these simple examples is qualitatively similar to all the other (more interesting) examples that we have studied. 

The effective channels for the repetition code with odd $n$ were first calculated by \citet{Greenbaum2016}. 
We present a slightly streamlined analysis and apply it to both the case of even $n$ and the conditional noise processes $\mathcal{N}_{\vec{s}}$. 
We will also see that our techniques can be used to find the effective channel for more general codes and to obtain our main analytical result which is a bound on the diamond norm error after ideal correction for unitary errors and general stabilizer $[n,k,d]$ codes. 

Our approach involves considering the case of unitary errors initially. 
We are therefore motivated to consider correcting errors of the form $U=\exp (-i\theta H)$ using an $[n,k,d]$ error correcting code. 
So we have $n$ physical qubits each undergoing an error $U$. 
We will confine ourselves to identically distributed errors for now. 
Since $U=\cos \theta I - i\sin\theta H$ and $H^2=I$ we can think of the fundamental error process as $H$ and each qubit will either have experienced an error $H$ or not. 
We can specify all possible tensor products of $H$ by a binary vector $\vec{w}$, with ones indicating that $H$ acts on the corresponding qubit.  
The Hamming weight $w$ is just the number of ones $|\vec{w}|$ in $\vec{w}$. 
Therefore $w$ is the number of qubits that have a non-trivial error.  
We can indicate error configurations by the operators $H_{\vec{w}}=H^{\otimes \vec{w}}$. 
There are $\binom{n}{w}$ configurations of weight $w$, and we can define the sum over these error configurations as follows $E_w=\sum_{\vec{w}:|\vec{w}|=w} H_{\vec{w}}$. 
So the overall $n$-qubit error operator is as follows 
\begin{equation}
\tilde{U}= [\exp (i\theta H)]^{\otimes n}= [\cos \theta I - i\sin\theta H]^{\otimes n} = \sum_{w=0}^n(-i\sin\theta)^w  (\cos\theta)^{n-w}  E_w. \label{eq:noise}
\end{equation}
This expression for the overall unitary error as a sum over Hamming weights holds for any independent unitary error regardless of the error correcting code of interest. 
In the following it will be convenient to define the function
\begin{equation}
f_w(\theta)=(-i\sin\theta)^w  \cos^{n-w}(\theta) .
\end{equation}
Notice that these functions are increasingly small for high weight. 
It is easy to see that $\lvert f_w(\theta)\rvert \leq \lvert\sin\theta\rvert^w\lvert\cos(\theta)\rvert^{n-w}\leq D_\diamond^w$ where $D_\diamond$ is the diamond norm error of the unitary noise.

At this point we have characterized the noise process $\tilde{\mathcal{N}}$ on the $n$ physical qubits. 
The next step is to calculate the effect of error recovery $\mathcal{R}$. 
This will involve measuring the syndrome, and performing a corresponding correction. 
As a result $\mathcal{R}$ involves a sum over the $2^{n-k}$ syndromes and it will not be possible to find a closed form expression for a general code. 
In this work we have used automated procedures to handle examples with small $k$. 
Moreover, the codes of most interest possess high symmetry, and this enhanced symmetry can be used to find closed form expressions in a number of interesting cases. 

In order to determine the syndrome that will arise from stabilizer measurement, we now replace the sum over error weights in \cref{eq:noise} by a sum over syndromes. 
At this point we have to take advantage of the specific properties of the code. 
We will focus for the moment on a repetition code with an odd number $n=2t+1$ of physical qubits. 
The stabilizer generators of the code are $X_iX_{i+1}$ for $i=1,...n-1$. 
There is one encoded qubit with logical operators $\bar{X}=X^{\otimes n}$ and $\bar{Z}=Z^{\otimes n}$. \acdedit{Recall that logical operators of stabiliser codes are defined only up to multiplication by stabilisers. So we could equally have chosen $\bar{X}$ to be $X\otimes I^{\otimes n-1}$, but we make this choice to emphasize the permutation symmetry of the code.}
While this is only technically a quantum error correcting code, as it has distance $d=1$, it does have distance $n$ for $Z$-errors. 
Consequently it should be effective in correcting unitary errors when $H=Z$. 
We will now specialize to that case. 
Notice that this corresponds to our noise model \cref{eq:noisemodel} with $p=0$, or equivalently $x=\sin^2\theta$ and $y=\sin\theta\cos\theta$.

Since the code can correct $t$ $Z$-errors, each of the $2^{n-1}$ error configurations $Z_{\vec{w}}$ with weight less than or equal to $t$ results in a unique syndrome. 
For each error configuration $\vec{w}$ with weight $w\leq t$ the weight $n-w$ error configuration $Z_{\vec{w}}\bar{Z}$ results in the same syndrome. 
Together these account for all $2^n$ error configurations. 
We can rewrite the error unitary (\ref{eq:noise}) as follows
\begin{equation}
\tilde{U}= \sum_{w=0}^t \sum_{\vec{w}:|\vec{w}|=w} Z_{\vec{w}}[f_w(\theta)\bar{I}+f_{n-w}(\theta)\bar{Z}]\label{eq:syndromes}.
\end{equation}
Each term in the sum over $\vec{w}$ results in a unique syndrome and so the syndrome measurement step projects onto a single value of $\vec{w}$. 
The vector $\vec{w}$, which has weight less than or equal to $t$, specifies the correction $Z_{\vec{w}}$. 
Consequently when the syndrome requires a correction of weight $w$ the logical qubit undergoes the effective channel
\begin{equation}
\mathcal{N}_{Lw}(\rho_L) = K_w \rho_L K_w^\dagger \label{eq:conditionalchannel}
\end{equation}
where 
\begin{equation}
K_w= f_w(\theta)\bar{I}+f_{n-w}(\theta)\bar{Z}. \label{eq:Kraus}
\end{equation}
It is straightforward to show that $K_w^\dagger K_w= \bigl(|f_w|^2+|f_{n-w}|^2\bigr)\bar{I}$ and therefore $K_w$ is proportional to a unitary. 
The normalization relates to the probability of the syndrome, and for small $\theta$ this is approximately $\theta^{2w}$. 
The unitary is a rotation about $Z$ and the rotation angle is approximately $(-1)^{t-w}\theta^{2(t-w)+1}$ in this limit. 
Notice that this rotation angle is much larger for the less likely syndromes; we will discuss this further in a subsequent section. 
The fact that the effective channel only depends on the weight of the correction is a consequence of the permutation symmetry of the code. 
For general noise processes, the effective channel conditioned on a given syndrome is analyzed by \citet{Chamberland2017} as a generalization of the approach of~\citet{Rahn2002}.

We can now evaluate the overall noise map as follows
\begin{equation}
\mathcal{N}_{L}(\rho) = \sum_{w=0}^t \binom{n}{w} K_w \rho_L K_w^\dagger =(1-x')\rho + x'\bar{Z}\rho \bar{Z} - iy'(\bar{Z}\rho - \rho \bar{Z}), \label{eq:newmap}
\end{equation}
with
\begin{align}
    x' &= \sum_{w=0}^{t}\binom{n}{w}x^{n-w}(1-x)^w =\sum_{w'=t+1}^{n}\binom{n}{w'}x^{w'}(1-x)^{n-w'} \\
    y' &=y^n \sum_{w=0}^t \binom{n}{w}(-1)^{t-w} =\binom{2t}{t}y^{2t+1}. \label{eq:telescope}
\end{align}
The final expression for $x'$ is a simple reorganization of the sum to make it clear that $x'\simeq \binom{n}{t+1}x^{t+1}$ for small $x$. 
It is clear that for small $x$ the sum for the effective error rate $x'$ is dominated by the terms arising from syndromes with $w=t+1$, the lowest weight uncorrected errors. 
(The repetition code corrects $t$ errors with $t\simeq n/2$ so the combinatorial factor $\binom{n}{t+1}$ attains its maximum possible value. 
By contrast, for surface codes with large distance, for example, the lowest weight errors do not necessarily dominate the effective error rate for realistic values of the parameters.) 
The final expression for $y'$ arises from applying Pascal's rule to consecutive terms in the sum~\cite{Greenbaum2016}. 
It is clear from this expression that in contrast to $x'$, all Hamming weights contribute to the expression for $y'$ at the same order $y^{2t+1}$. 
Thus any characterization of the average error will get contributions from all syndromes, even those that correspond to very high weight errors and arise with very low probability. 
However the combinatorial factors in the sum in \cref{eq:telescope} means that the sum is largely determined by terms where $w$ is not too different from $t$. 

\acdedit{Inpsecting the sum in \cref{eq:telescope}, it is clear that the value of $y'$ is reduced due to the fact that the sign of the rotation angle of the effective channel for each syndrome oscillates with the Hamming weight. By reducing the magnitude of $y'$ this cancellation reduces the coherence of the effective channel. It is straightforward to assess the extent to which this cancellation is reducing the coherence of the average channel. If we remove the factor $(-1)^{t-w}$ from the sum in \cref{eq:telescope}  we would obtain the value $2^ny^{2t+1}$.  
So this quantity scales with $y$ as before but it is larger than $y'$ by a factor roughly equal to $\sqrt{\pi t}$ in the limit of large $t$.}

We have established this identity only for the unitary case where $x=\sin^2\theta$ and $y=\sin\theta\cos\theta$. 
However we will show below that this formula holds for all the combinations of unitary and dephasing errors in \cref{eq:noisemodel}. 
Notice that since $x'$ depends on $x$ alone, the \acdedit{Pauli Twirling Approximation} to the effective channel can be found starting from the \acdedit{Pauli Twirling Approximation} at the physical level. 
This seems unlikely to be a general property of these noise maps.

We can also write the conditional channel (\ref{eq:conditionalchannel}) for a given syndrome $\vec{w}$. 
Note that
\begin{equation}
\mathcal{N}_{L\vec{w}}(\rho_L) =  K_w \rho_{L}K_w^\dagger =\bar{x}_\vec{w}'\rho_L + x_\vec{w}'\bar{Z}\rho_L \bar{Z} - iy_\vec{w}'(\bar{Z}\rho_L - \rho_L \bar{Z}) \label{eq:conditionalchannel2}
\end{equation}
where
\begin{align}
\bar{x}'_\vec{w} &= x^{w}(1-x)^{n-w}  \\
    x'_\vec{w} &= x^{n-w}(1-x)^w\\
    y'_\vec{w} &=(-1)^{t-w}y^n .
\end{align}
Note that all of these are proportional to unitary processes satisfying $y'^2_\vec{w}=x_\vec{w}'\bar{x}'_\vec{w}$ since the underlying process is unitary and has $y^2=x(1-x)$.

The overall probability of this syndrome is
\begin{equation}
p_\vec{w}= x^w(1-x)^{n-w}+x^{n-w}(1-x)^{w} ,
\end{equation}
and there are $\binom{n}{w}$ syndromes that have a correction operation of weight $w$. 
Since the probability of each weight depends on $x$ alone, it can be computed using the Pauli approximation to \cref{eq:noisemodel}. 
Also the probability of a given weight arising is independent of the logical state. 
We will see that these properties do not hold for all codes. 

As for the unconditional effective channel, we have so far explained how to establish these results with unitary errors $p=0$, but we have rewritten these expressions in terms of the parameters $x$ and $y$ such that they hold for all instances of the noise model \cref{eq:noisemodel}. 
We will establish these formulas in the general case in the next subsection. 

\subsection{Effective channels with dephasing as well as unitary errors} \label{sec:fullmodel}

Given the calculations so far it is straightforward to see that for the general case of the noise model of \cref{eq:noisemodel} we have
\begin{equation}
\tilde{\mathcal{N}}(\rho)= \sum_{\vec{w}',\bar{\vec{w}},\bar{\vec{w}}'} p^{w'}(1-p)^{n-w'}f_{\bar{w}}(\theta)f^*_{\bar{w}'}(\theta)Z_{\vec{w}'+\bar{\vec{w}}}\rho Z_{\vec{w}'+\bar{\vec{w}}'}.
\end{equation}
Here and elsewhere addition of the binary vectors $\vec{w}$ is modulo two.

If we consider a syndrome measurement with outcome $\vec{w}$ then a projection onto the corresponding stabilizer subspace just picks out the four terms of this sum where both $\vec{w}'+\bar{\vec{w}}$ and $\vec{w}'+\bar{\vec{w}}'$ are equal to either $\vec{w}$ or $\vec{n}-\vec{w}$. 
($\vec{n}$ is the $n$-component binary vector with Hamming weight $n$.) The corresponding correction operator is $Z_{\vec{w}}$ as before, and this results in a conditional effective channel of the form of \cref{eq:conditionalchannel2} as expected. 

We now just need to evaluate the parameters $\bar{x}'_\vec{w},x'_\vec{w},y'_\vec{w}$. 
Considering first the coefficient of $\rho_L$ we find
\begin{eqnarray}
\bar{x}'_\vec{w}&=& \sum_{\vec{w}',\bar{\vec{w}}=\vec{w}-\vec{w}'} p^{w'}(1-p)^{n-w'}|f_{\bar{w}}(\theta)|^2 \nonumber \\& =&\sum_{w'=0}^n\sum_{\tilde{w}=0}^{w'}\binom{w}{\tilde{w}}\binom{n-w}{w'-\tilde{w}} p^{w'}(1-p)^{n-w'}|f_{w+w'-2\tilde{w}}(\theta)|^2 \nonumber \\
&=& x^w(1-x)^{n-w}. \label{eq:computation}
\end{eqnarray}
In the first equality we sum over all binary vectors $\vec{w}'$ and $\bar{\vec{w}}=\vec{w}-\vec{w}'$. 
Consider a fixed $\vec{w}'$. 
Let $\tilde{w}$ be the number of locations where both $\vec{w}$ and $\vec{w}'$ have ones. 
There are $w'-\tilde{w}$ locations where $\vec{w}'$ has a one and $\vec{w}$ has a zero. 
Consequently $\bar{\vec{w}}$ has weight $|\bar{\vec{w}}|=w-\tilde{w}+(w'-\tilde{w}) = w+w'-2\tilde{w}$. 
For each $0\leq \tilde{w}\leq w'$ there are $\binom{w}{\tilde{w}}\binom{n-w}{w'-\tilde{w}}$ distinct choices of $\vec{w}'$ that have this value of $\tilde{w}$. 
This establishes the second equality above. 
The final equality is most easily seen by observing that $x=(1-p)\sin^2\theta+p\cos^2\theta$ and $1-x=(1-p)\cos^2\theta+p\sin^2\theta$ and then applying the binomial expansion to $x^w$ and $(1-x)^{n-w}$.

The same procedure can be used to compute the other parameters. 
The case of $x'_\vec{w}$ is precisely the same as for $\bar{x}'_\vec{w}$. 
For $y'_\vec{w}$ we find
\begin{eqnarray}
-i y'_\vec{w}&=& \sum_{\vec{w}',\bar{\vec{w}}=\vec{n}-\vec{w}-\vec{w}',\bar{\vec{w}}'=\vec{w}-\vec{w}'} p^{w'}(1-p)^{n-w'}f_{\bar{w}}(\theta) f^*_{\bar{w}'}(\theta) \nonumber \\& =&\sum_{w'=0}^n\sum_{\tilde{w}=0}^{w'}\binom{w}{\tilde{w}}\binom{n-w}{w'-\tilde{w}} p^{w'}(1-p)^{n-w'}f_{n+2\tilde{w}-w-w'}(\theta) f^*_{w+w'-2\tilde{w}}(\theta) \nonumber \\
& =& -i(-1)^{t-w}[\sin\theta\cos\theta]^n\sum_{w'=0}^n\sum_{\tilde{w}=0}^{w'}\binom{w}{\tilde{w}}\binom{n-w}{w'-\tilde{w}} (-p)^{w'}(1-p)^{n-w'} \nonumber \\
&=& -i(-1)^{t-w}y^n. 
\end{eqnarray}
One way to see the final equality, is to note that $y=\cos\theta \sin\theta[(1-p)-p] $ and apply the binomial expansion to $y^w$ and $y^{n-w}$. 

This establishes the claimed result for the conditional effective channels. 
Averaging over syndromes recovers the result for the unconditional effective channel.

\subsection{Repetition codes with even distance}

So far we have considered the case where $n=2t+1$ is odd. 
It is also interesting to look at the case of an even distance code where $n=2t+2$ physical qubits. 

We will note two interesting features of this case. 
First, the unitary part of the logical error channel ($y'$) vanishes after a single round of error correction, so that the effective logical channel is pure dephasing. 
Second, the probabilities of various syndrome measurement outcomes will depend on the initial logical state (albeit weakly). 

The stabilizer generators of the code are $X_iX_{i+1}$ for $i=1,...n-1$. 
There is one encoded qubit with logical operators $\bar{X}=X\otimes I^{\otimes n-1}$ and $\bar{Z}=Z^{\otimes n}$. 
For the odd repetition codes, each syndrome indicates either a weight $w$ or a weight $n-w$ error and the symmetric decoder corrects the lowest weight error. 
When $n$ is even there is no good way of correcting the Hamming weight $t+1$ errors, and we just need to decide on some procedure that doesn't lower the symmetry of the code. 
In this case for each of the $w=t+1$ syndromes we will choose one of the two corresponding weight $t+1$ errors as the correction operation. 
So for example one of the pair will act nontrivially on the first qubit, and we will choose this as the correction operator. 
This procedure results in a state in the code space and preserves the overall structure and symmetry of the decoder. 
Looking at \acdedit{the} sum over the syndromes in \cref{eq:syndromes} we should therefore treat the $w=t+1$ term separately as follows:  
\begin{equation}
\tilde{U}= \sum_{w=0}^t \sum_{\vec{w}:|\vec{w}|=w} Z_{\vec{w}}[f_w(\theta)\bar{I}+f_{n-w}(\theta)\bar{Z}] +\sum_{\vec{w}:|\vec{w}|=t+1, w_1=1} f_{t+1}(\theta)Z_{\vec{w}}[\bar{I}+\bar{Z}] \label{eq:syndromeseven}.
\end{equation}

It is already clear from the expression for the noise operator in terms of the syndromes that the effective channel conditioned on a stabilizer measurement is quite different in the even distance case as compared to the odd distance case. 
We can read off from \cref{eq:syndromeseven} that a weight $t+1$ syndrome results in a projection on a $\bar{Z}$ eigenstate of the logical qubit, which would of course totally destroy the quantum information. 
In contrast for odd distance codes the conditional channel was unitary. 
For the other syndromes we have as before a conditional channel given by \cref{eq:conditionalchannel} with the Kraus operator $K_w$ given by \cref{eq:Kraus}. 
However in the even distance case it is easy to show that the Kraus operator is not unitary and in fact $K_w^\dagger=\pm K_w$. 
For small $\theta$ we find $K_w\approx (-i)^w\theta^w [\bar{I}-(-1)^{t-w}\theta^{2(t-w)+2} \bar{Z}]$. 
Each of these Kraus operators corresponds to performing a weak measurement of $\bar{Z}$ with a measurement outcome that depends on the weight of the stabilizer. 

An important takeaway is this: 
whereas previously the syndrome probabilities were completely independent of the logical state, this is no longer true for the even distance codes. 
For these codes the unitary error is generating some entanglement between the logical and stabilizer qubits and the stabilizer measurement outcomes can depend on the initial logical state.

We can now evaluate the overall noise map for even distance repetition codes, which has the same form as \cref{eq:newmap} but with the parameters as follows
\begin{align}
    x' &= \sum_{w=0}^{t}\binom{n}{w}x^{n-w}(1-x)^w + \frac{1}{2}\binom{n}{t+1}x^{t+1}(1-x)^{t+1}\\ &= \frac{1}{2}\binom{n}{t+1}x^{t+1}(1-x)^{t+1} +\sum_{w'=t+2}^{n}\binom{n}{w'}x^{w'}(1-x)^{n-w'} \\
    i y' &=y^n \left[\sum_{w=0}^t \binom{n}{w}(-1)^{t-w} - \frac{1}{2} \binom{n}{t+1} \right]=0.
\end{align}
As we argued in the odd case this expression holds for all initial $x$ and $y$. 
A cancellation arises such that the unitary component of the error is totally removed after a single round of error correction. 
Once again this cancellation depends on the fact that all Kraus operators contribute at the same order, so in practice for a large $t$ code this involves averaging over very rare events. 

As for the case of odd distance repetition codes, it is possible to generalize the the conditional effective channels to the general noise model of \cref{eq:noisemodel}. 
We won't write the formulas for $\bar{x}'_w,x'_w,y'_w$ explicitly. 
However note that the overall probability for a syndrome $\vec{w}$ is 
\begin{equation}
p_\vec{w}= x^w(1-x)^{n-w}+x^{n-w}(1-x)^{w}-2(-1)^{t-w}y^n\langle \bar{Z}\rangle.
\end{equation}
In contrast to the case of odd distance, this probability depends on $y$ as well as $x$ and so the behavior is distinct from that for the Pauli channel approximation. 
It also depends on the logical state. 
Measuring the syndrome implements a weak measurement of $\bar{Z}$ and this is reflected in the syndrome outcome probabilities. 
Note that for low $w$ this dependence is rather weak.

\subsection{Non-uniform single-qubit errors} \label{sec:nonuniform}

These considerations work equally well when the unitary error on each physical qubit can be distinct. 
The generalization from the case of uniform unitary errors can largely be handled by modifying the notation. 
We will consider the case of rotations about the $z$-axis by angles $\theta_j$ that differ between the qubits.

Stepping through the calculations above to find the effective channel for a given syndrome we first write the error process on the $n$ qubits as a sum over syndromes as follows
\begin{equation}
\tilde{U} =\bigotimes_{i=1}^n U(\theta_i) = \sum_{w=0}^t \sum_{\vec{w}:|\vec{w}|=w} Z_{\vec{w}}[f_\vec{w}(\bm{\theta})\bar{I}+f_{\vec{n}-\vec{w}}(\bm{\theta})\bar{Z}].
\end{equation}
We have defined the following function 
\begin{equation}
f_\vec{w}(\bm{\theta})=\prod_{j=1}^n(-i\sin\theta_j)^{w_j}  \cos^{1-w_j}(\theta_j) .
\end{equation}
Consequently the Kraus operator corresponding to the syndrome $\vec{w}$ is as follows
\begin{equation}
K_\vec{w}= f_\vec{w}(\bm{\theta})\bar{I}+f_{\vec{n}-\vec{w}}(\bm{\theta})\bar{Z}. \label{eq:Kraus2}
\end{equation}
Now that the unitary errors are not uniform, the effective channels for each syndrome are distinct. 
The number of distinct effective channels has risen from $t$ to $2^{n-1}$. 

We can write these effective channels in terms of the parameters $x_j$ and $y_j$ as follows
\begin{align}
\bar{x}'_\vec{w} &= \prod_{j=1}^n x_j^{w_j}(1-x_j)^{1-w_j}  \\
    x'_\vec{w} &= \prod_{j=1}^n x_j^{1-w_j}(1-x_j)^{w_j}  \\
    y'_\vec{w} &=(-1)^{t-w}\prod_{j=1}^ny_j .
\end{align}

As in our earlier examples these expressions have been established only for unitary rotations about the $z$ axis but hold also for the model with non-uniform dephasing included. 
In this case the generalization can be established more straightforwardly. 
The map from the physical noise process, specified by the parameters $\{x_j\}$, $\{y_j\}$ and $\{\bar{x}_j=1-x_j\}$ to the effective logical noise process for the syndrome $\vec{w}$, given by $\bar{x}'_\vec{w} ,x'_\vec{w}, y'_\vec{w} $ is a polynomial where each term is $n$-th order overall and first order in $\bar{x}_j,x_j,y_j$ for each $j$. 
These polynomials are fixed uniquely by considering just the unitary case. 
We were not able to use this argument previously because in the unitary case $y^2=\bar{x}x$ which generally leads to an ambiguity in generalizing from the unitary to the general case using this method. 
This ambiguity is avoided here because no quadratics in the parameters $\bar{x}_j,x_j,y_j$ arise. 
Having established these results for non-uniform noise without explicitly considering dephasing we could re-obtain the results of \cref{sec:fullmodel} by specializing to the uniform case and then averaging over syndromes.

\section{Effective channels with coherent errors for general stabilizer codes}\label{sec:theoremsection}

The same approach can be used to find effective channels for more general stabilizer codes, including the channel conditioned on a given syndrome. 
The main ideas have already arisen in the context of the repetition codes, but for a general stabilizer code there is a greater overhead of notation and some technical details. 

The main result of this subsection is that the diamond distance error for unitary noise processes after ideal correction scales like $\theta^d$ where $d$ is the distance of the code when a minimum distance decoder is used. 
Readers who are more interested in the specific effective channels may wish to skip ahead to the next section.

As discussed in the appendix, a general product of Pauli matrices on $n$ qubits can be specified by a binary vector $\vec{b}$ of length $2n$. 
Ones in the first $n$ entries specify the locations of $X$ operators, while the ones in the second group of $n$ entries specify the locations of $Z$ operators. 
A qubit acted on by both an $X$ and a $Z$ is understood to be acted on by $Y$. 
We will use the notation $P_{\vec{b}}$ to indicate the resulting product of Paulis. 
If $m_x$, $m_y$ and $m_z$ are the numbers of $X$, $Y$ and $Z$ operators in the product, then the Hamming weight is $w=m_x+m_y+m_z$. 
We refer the reader to \cite{Nielsen2000} for further information on stabilizer quantum error correcting codes. 
The discussion in the appendix provides a brief summary and fixes our notation.  

We now wish to consider general unitary rotations, so we have $H=\alpha_x X + \alpha_y Y + \alpha_z Z$. 
Since $H^2=I$ we have $\alpha_x^2+\alpha_y^2+\alpha_z^2=1$. 
We can therefore rewrite our overall unitary on $n$ qubits (\ref{eq:noise}) as follows
\begin{equation}
\tilde{U}=\sum_{\vec{b} } f_w(\theta)\alpha^\vec{m} P_\vec{b} \label{eq:fullerror}
\end{equation}
where we have used the shorthand notation $\alpha^\vec{m} =\alpha_x^{m_x}\alpha_y^{m_y}\alpha_z^{m_z}$.

As before we can generalize to the case of non-uniform unitary noise straightforwardly. 
We obtain
\begin{equation}
\tilde{U}=\sum_{\vec{b} } g_\vec{b}(\bm{\theta},\bm{\alpha}) P_\vec{b},
\end{equation}
with the definition
\begin{equation}
g_\vec{b}(\bm{\theta},\bm{\alpha})=\prod_{j=1}^n(-i\sin\theta_j)^{w_j(\vec{b})}  \cos^{1-w_j(\vec{b})}(\theta_j)\alpha_{xj}^{m_{xj}(\vec{b})}\alpha_{yj}^{m_{yj}(\vec{b})}\alpha_{zj}^{m_{zj}(\vec{b})} .
\end{equation}
Both the rotation angles $\theta_j$ and the rotation axes, given by $\alpha_{xj},\alpha_{yj},\alpha_{zj}$, now change from qubit to qubit. 
We have used a notation such that $m_{xj}(\vec{b})$ is one if $P_\vec{b}$ has an $X$ for qubit $j$ and is zero otherwise.

Consider now a stabilizer code. 
The stabilizer generators are signed products of Pauli matrices $\pm P_{\vec{g}_i}$  where the binary vectors $\vec{g}_i$ satisfy certain constraints, for example that the stabilizer generators all commute. 
A general member of the stabilizer group $S_{\vec{t}}$ is described by the length $n-k$ binary vector $\vec{t}$, where ones in the vector signify that the corresponding stabilizer generator is part of the product that results in $S_{\vec{t}}$. 
The code has $k$ logical $X$ operators $\bar{X}_i= P_{\vec{x}_i}$ and $k$ logical $Z$ operators $\bar{Z}_i= P_{\vec{z}_i}$. 
These logical operators are defined only up to multiplication by elements of the stabilizer group, but here we will pick particular coset representatives. 
We can describe products of logical Pauli operators $L_{\vec{l}}$  by a length $2k$ binary vector $\vec{l}=(\vec{l}_x,\vec{l}_z)$, in analogy to the Pauli matrices. 

Now consider implementing the error correction with a specific decoder. 
First the stabilizers are measured, and there are $2^{n-k}$ syndromes. 
For a given Pauli error $P_{\vec{b}}$, we associate the syndrome $\vec{y} = \vec{s}(\vec{b})$. 
We will choose a fiducial Pauli error $E_\vec{y}$ having the lowest possible weight among all errors that lead to that syndrome. 
There may be more than one choice for some syndromes, particularly for degenerate codes, and in that case we choose the lexicographically first choice of Pauli. 
We therefore define the \textbf{symmetric decoder} (or \textbf{minimum weight decoder}) to be the decoder that chooses the following correction, 
\begin{align}
	E^\dagger_\vec{y} = P_{\vec{r}(\vec{y})}, \qquad \vec{r}(\vec{y}) = \mathrm{lex} \arg\min_{\vec{b}} \Bigl\{  |\vec{b}| : \vec{s}(\vec{b})=\vec{y}\Bigr\},
\end{align}
where recall that $|\vec{b}|$ is the Hamming weight of $\vec{b}$. We note that this decoder is generally not optimal, so some decoders could potentially do even better than suggested by our analysis.

Later in the paper we will discuss two slight variants of this decoder. 
In \cref{sec:symmetric} we consider simplifications to the effective noise channels that arise from symmetries of the code. 
In that case choosing the correction to be the lexicographically first choice of Pauli with the required syndrome and lowest possible weight can result in a correction procedure that has less symmetry than the code itself. 
An example is the Steane code, where certain two qubit errors can be corrected. 
If these are chosen to be the two-qubit Paulis made up of one $X$ and one $Z$ error then the effective noise channels simplify. 
In such examples it is possible to specify the corrections $E_{\vec{y}}$ out of the set of possible lowest weight corrections for that syndrome so as to preserve the symmetry of the code. 
A second modification of the decoder is used in \cref{sec:concatcodeperformance} where we consider noise processes that are combinations of dephasing and $Z$-rotations. 
Since for these noise processes only $Z$ errors can occur, we can improve the performance of the code by choosing $E_\vec{y}$ to be the lowest weight product of Pauli $Z$ errors that could produce the syndrome $\vec{y}$. 
If no product of $Z$s results in the syndrome we can, as before, choose the lowest weight correction, however this syndrome will not arise.

Having specified a decoder we know that for any Pauli matrix $P_\vec{b}$ there is a unique syndrome $\vec{s}(\vec{b})$ and correction $E^\dagger_{\vec{s}(\vec{b})}$ such that $E^\dagger_{\vec{s}(\vec{b})}P_\vec{b}$ commutes with the stabilizer group. 
Thus $E^\dagger_{\vec{s}(\vec{b})}P_\vec{b}$ is in the normalizer group that is generated by the logical operators and the stabilizer group, see for example~\cite{Nielsen2000}. 
Thus there exists a logical operator $L_{\vec{l}(\vec{b})}$, stabilizer $S_{\vec{t}(b)}$, and phase factor $\eta_{\vec{b}}$ equal to $\pm1$ or $\pm i$ such that $E^\dagger_{\vec{s}(\vec{b})}P_\vec{b}=\eta_\vec{b}L_{\vec{l}(\vec{b})}S_{\vec{t}(\vec{ b})}$. 
Consequently every Pauli matrix can be written uniquely in the form $P_\vec{b}=\eta_\vec{b}E_{\vec{s}(\vec{b})}L_{\vec{l}(\vec{b})}S_{\vec{t}(\vec{ b})}$. 
This mapping is a property of the chosen error correcting code, encoding unitary and decoding scheme. 
Equally for any choice of syndrome $\vec{s}$, logical operator $\vec{l}$, and stabilizer element $\vec{t}$ their product is a Pauli matrix up to a phase so that there is a unique $\vec{b}(\vec{s},\vec{l},\vec{t})$ and $\eta_{\vec{s},\vec{l},\vec{t}}$ such that $P_\vec{b}=\eta_{\vec{s},\vec{l},\vec{t}}E_{\vec{s}}L_{\vec{l}}S_{\vec{t}}$. 
So the mapping between $\vec{b}$ and $(\vec{s},\vec{l},\vec{t})$ is one-to-one and onto.
Note that the $4^n$ Pauli products $P_\vec{b}$ are accounted for since there are $2^{n-k}$ syndromes, $4^k$ logical operators and $2^{n-k}$ stabilizer elements, which can together result in $2^{n-k}\cdot 4^k \cdot 2^{n-k}=4^n$ distinct products.

With these preliminaries completed, we can now rewrite the expression for $\tilde{U}$ from \cref{eq:fullerror} as a sum over syndromes, rather than a sum over products of Paulis: 
\begin{equation}
\tilde{U}= \sum_{\vec{s}} \sum_{\vec{l}}\sum_{\vec{t}} f_{w(\vec{s},\vec{l},\vec{t})}(\theta)\alpha^{\vec{m}(\vec{s},\vec{l},\vec{t})}\eta_{\vec{s},\vec{l},\vec{t}}E_{\vec{s}}L_{\vec{l}}S_{\vec{t}} \label{eq:fullsyndrome}.
\end{equation}

This expansion for $\tilde{U}$ makes it very straightforward to read off the Kraus operators of the effective channel. 
Since we start in the code space the stabilizers in this product act trivially. 
The syndrome measurement projects the system onto a single value of $\vec{s}$ and the correction $E_{\vec{s}}^\dagger$ removes the fiducial error for that syndrome. 
Consequently we obtain $2^{n-k}$ Kraus operators indexed by $\vec{s}$ acting on the logical qubit as follows
\begin{equation}
K_{\vec{s}}=  \sum_{\vec{l}}\sum_{\vec{t}} f_{w(\vec{s},\vec{l},\vec{t})}(\theta)\alpha^{\vec{m}(\vec{s},\vec{l},\vec{t})}\eta_{\vec{s},\vec{l},\vec{t}}L_{\vec{l}}. \label{eq:generalKraus}
\end{equation}
The non-uniform case is nearly identical
\begin{equation}
K_{\vec{s}}=  \sum_{\vec{l}}\sum_{\vec{t}}  g_{\vec{b}(\vec{s},\vec{l},\vec{t})}(\bm{\theta},\bm{\alpha})   \eta_{\vec{s},\vec{l},\vec{t}}L_{\vec{l}}. \label{eq:nonuniformKraus}
\end{equation}

This expression can be used to study general properties of the effective channel. 
We will restrict our discussion to the following theorem. 

\begin{theorem}
\label{thm:theorem}
For any $[n,k,d]$ stabilizer code, the diamond norm error $D'_\diamond$ of the logical qubits after ideal error correction using a minimum weight decoder satisfies
\begin{equation}
D'_\diamond \leq c_{n,k} D^d_\diamond
\end{equation}
for independent unitary errors with rotation angles satisfying $|\theta_i| \le \pi/4$, where $D_\diamond = \max_{i} \lvert\sin\theta_i\rvert$. 
The constant $c_{n,k}$ depends on $n$ and $k$ but is independent of the errors. 
\end{theorem}

\begin{proof}
Our result is analogous to Theorem 5.5 of~\cite{Fern2006} when applied to unitary errors, and our proof follows the contours of that theorem. 
We will explicitly consider the case of uniform errors. 
The generalization to non-uniform errors is a straightforward modification using the approach of \cref{sec:nonuniform} and \cref{eq:nonuniformKraus} for the Kraus operator arising from the syndrome $\vec{s}$ with non-uniform unitary noise.

We must first explicitly separate the nontrivial logical operators from the identity in the expansion over syndromes. 
So consider
\begin{equation}
\tilde{U}= \sum_{\vec{s}} \sum_{\vec{t}} f_{w(\vec{s},\vec{0},\vec{t})}(\theta)\alpha^{\vec{m}(\vec{s},\vec{0},\vec{t})}\eta_{\vec{s},\vec{0},\vec{t}}E_{\vec{s}}\bar{I}S_{\vec{t}}+\sum_{\vec{s}} \sum_{\vec{l}\neq \vec{0}}\sum_{\vec{t}} f_{w(\vec{s},\vec{l},\vec{t})}(\theta)\alpha^{\vec{m}(\vec{s},\vec{l},\vec{t})}\eta_{\vec{s},\vec{l},\vec{t}}E_{\vec{s}}L'_{\vec{l}}S_{\vec{t}}.\label{eq:byweight}
\end{equation}
We have written $L'_{\vec{l}}$ in the second term above just to emphasize that the logical identity operator does not arise in the sum. 
We seek to understand the relative size of the terms in this expansion. 
We have restricted attention to $|\theta|\leq \pi/4$ so that higher weight errors are less likely than lower weight errors. 
Notice that in this regime we have $\lvert\tan\theta\rvert\leq 1$ and therefore $|f_w(\theta)|\leq |f_{w'}(\theta)|$ when $w\geq w'$. 
On the other hand all the factors $\alpha^\vec{m}$ satisfy $|\eta_{\vec{b}}\alpha^\vec{m}|\leq 1$ and the operator norm of the product of Pauli matrices $E_{\vec{s}}L_{\vec{l}}S_{\vec{t}}$ is $\leq 1$.

We now analyze the Hamming weights of the various terms in \cref{eq:byweight}. 
In each term of this sum we will specify the weight of the correction operator $E_{\vec{s}}^\dagger$ by $w_s$. 
Considering the first of the two terms in \cref{eq:byweight}, the identity is a member of the stabilizer group, so that one of the terms is just $E_\vec{s}$ itself and has Hamming weight $w_s$. 
The other terms in the sum are of the form $E_\vec{s} S_{\vec{t}}$ and each of these terms is a possible error process that leads to the same syndrome $\vec{s}$. 
But we have chosen our decoder such that $w_s$ is the lowest possible weight for an error with this syndrome. 
So all the contributions to the first term in \cref{eq:byweight} have $w\geq w_s$. 
Since we have a quantum error correcting code, all errors up to and including some Hamming weight $t$ will be corrected by the code. 
For odd $d$ we have $d=2t+1$, while for even $d$, $d=2t+2$. 
Every product of Paulis of Hamming weight $w\leq t$ is either a member of the stabilizer group, or one of the fiducial errors $E_{\vec{s}}$, or is equivalent to some $E_{\vec{s}}$ up to multiplication by a stabilizer operator. 
Consequently every product of Paulis with Hamming weight $w\leq t$ occurs somewhere in the first term of \cref{eq:byweight}. 
Considering now the second term in \cref{eq:byweight}, clearly we have $w\geq t+1$ for all the contributions, regardless of $\vec{s}$. 
For specific values of $\vec{s}$ we can be more precise. 
For each syndrome the factor $L'_{\vec{l}}S_{\vec{t}}$ is a non-trivial logical operator and therefore has weight at least $d$. 
Therefore in each contribution to the syndrome $\vec{s}$ the product of Pauli matrices $E_{\vec{s}}L'_{\vec{l}}S_{\vec{t}}$ has weight $w\geq d-w_s$. 
(Note that in examples like the Steane code it is possible to correct certain errors $E_\vec{s}$ with $w_s>t$. 
For these syndromes $w\geq t+1$ will be a better bound on the Hamming weight of $E_{\vec{s}}L'_{\vec{l}}S_{\vec{t}}$.)

Consequently we can consider the corresponding expansion of $K_\vec{s}$
\begin{equation}
K_{\vec{s}}=  \sum_{\vec{t}} f_{w(\vec{s},\vec{0},\vec{t})}(\theta)\alpha^{\vec{m}(\vec{s},\vec{0},\vec{t})}\eta_{\vec{s},\vec{0},\vec{t}}\bar{I} + \sum_{\vec{l}\neq \vec{0}}\sum_{\vec{t}} f_{w(\vec{s},\vec{l},\vec{t})}(\theta)\alpha^{\vec{m}(\vec{s},\vec{l},\vec{t})}\eta_{\vec{s},\vec{l},\vec{t}}L'_{\vec{l}}.
\end{equation}
In the first term we have $w\geq w_s$ and in the second we have $w\geq \max (d-w_s,t+1)$. 
The overall effective channel is
\begin{equation}
\mathcal{N}_L(\rho)=\sum_\vec{s} K_{\vec{s}}\rho K_{\vec{s}}^\dagger = \sum_{\vec{l},\vec{l}'} r_{\vec{l},\vec{l}'}L_{\vec{l}}\rho L_{\vec{l'}}^\dagger.
\end{equation}
Specifically
\begin{equation}
r_{\vec{l},\vec{l}'}=  \sum_{\vec{s},\vec{t},\vec{t}'} f_{w(\vec{s},\vec{l},\vec{t})}(\theta)f_{w(\vec{s},\vec{l}',\vec{t}')}(\theta)\alpha^{\vec{m}(\vec{s},\vec{l},\vec{t})}\alpha^{\vec{m}(\vec{s},\vec{l}',\vec{t}')}\eta_{\vec{s},\vec{l},\vec{t}}\eta_{\vec{s},\vec{l}',\vec{t}'}. \label{eq:coeffs}
\end{equation}
This is a trace-preserving, unital, completely positive map, since the encoding, recovery and decoding maps are all unital and trace preserving. 

If we specify that both $\vec{l}$ and $\vec{l}'$ are non-trivial so $L_\vec{l},L_{\vec{l}'}\neq \bar{I}$ then 
\begin{equation}
\label{eq:derivation1}
|r_{\vec{l},\vec{l}'}|\leq  \sum_{\vec{s},\vec{t},\vec{t}'} |f_{w(\vec{s},\vec{l},\vec{t})}(\theta)f_{w(\vec{s},\vec{l}',\vec{t}')}(\theta)|\leq  8^{n-k}|f_{t+1}(\theta)|^2\leq 8^{n-k}|f_{2t+2}(\theta)| \leq 8^{n-k}|f_{d}(\theta)|.
\end{equation}
The first inequality is just the triangle inequality combined with the fact that $|\alpha^\vec{m}|\leq 1$ and $\eta_{\vec{s},\vec{l},\vec{t}}\leq 1$. 
The second inequality follows from the fact that each $w$ appearing in the sum satisfies $w\geq t+1$ and so $|f_w|\leq |f_{t+1}|$ when $|\theta|\leq \pi/4$. 
The prefactor just counts the number of terms in the sum over the syndromes and two copies of the stabilizer group. 
To see the third inequality notice that if $w+w'\leq n$ then $|f_wf_{w'}|=\lvert\sin(\theta)\rvert^{w+w'}\lvert\cos(\theta)\rvert^{2n-w-w'}\leq \lvert\sin(\theta)\rvert^{w+w'}\lvert\cos(\theta)\rvert^{n-w-w'}=|f_{w+w'}|$. 

Applying the same logic to the case where $\vec{l}'=\vec{0}$, so that $L_{\vec{l}'}=\bar{I}$, and $\vec{l}\neq \vec{0}$ we see that
\begin{equation}
\label{eq:derivation2}
|r_{\vec{l},\vec{0}}|\leq  \sum_{\vec{s},\vec{t},\vec{t}'} |f_{w(\vec{s},\vec{l},\vec{t})}(\theta)f_{w(\vec{s},\vec{0},\vec{t}')}(\theta)|\leq  4^{n-k}\sum_{\vec{s}} |f_{w_s}(\theta)f_{d-w_s}(\theta)|\leq   8^{n-k}|f_{d}(\theta)|.
\end{equation}
The second inequality here results from the fact discussed above that for each term in the sum $w(\vec{s},\vec{l},\vec{t})\geq d-w_s$ and $w(\vec{s},\vec{0},\vec{t})\geq w_s$.
Clearly the case of $r_{\vec{0}\vec{l}'}$ with $\vec{l}'\neq \vec{0}$ is the same.

Finally we can bound $r_{\vec{0},\vec{0}}$ since $\mathcal{N}_L(\rho)$ is trace preserving and unital and so $\sum_{\vec{l}}r_{\vec{l},\vec{l}} = 1$, and therefore
\begin{equation}
|r_{\vec{0},\vec{0}}-1|\leq (4^k-1) 8^{n-k}|f_{d}(\theta)|.
\end{equation}

The result follows by applying the triangle inequality and the fact that if $\mathcal{N}_{\vec{l},\vec{l}'}(\rho)=L_{\vec{l}}\rho L_{\vec{l'}}^\dagger$ then $\|\mathcal{N}_{\vec{l},\vec{l}'}\|_\diamond = 1$. 
This last statement follows from the characterization of the diamond norm in terms of the maximum output fidelity in \cite[Theorem 5]{Watrous2013}. 
\end{proof}

A couple of remarks are in order. 
First, collecting the constant factors and simplifying, we see that the statement holds with $c_{n,k}=2^{3n+k} + 2^{3n-k} \le 2^{3n+k+1}$, but no attempt was made to optimize the constant and significantly better bounds might exist. 
Second, we note that most of the error correcting codes that we have studied saturate this inequality in the sense that the diamond norm distance after error correction scales as $|\theta|^d$ for a small initial unitary angle, regardless of code or rotation axis. 
Only for certain codes and specific rotation axes have we observed a more favorable scaling.
\acdedit{Third, our theorem bounds the diamond norm error in terms of $\theta^{d}$ for all stabilizer codes. However for large distance codes operating close to threshold this bound will be very weak. For large distance codes the expected $\theta^d$ dependence of the effective channel will occur but only for very small values of $\theta$, very far below threshold.}
We also note that the intermediate stages of our proof reproduce the observations about unitary channels made at the level of the process matrix by \citet{Gutierrez2016}. 
Finally, we note that a related version of this result in terms of the average gate infidelity was recently proven independently by~\citet{Beale2018}. 

\subsection{Effective channels for symmetric codes} \label{sec:symmetric}

It has been observed that the effective channels for widely studied error correcting codes, such as the five qubit and Steane codes, typically have many fewer distinct conditional channels than there are syndromes of the code~\cite{Chamberland2017}. 
We saw an example of this for the repetition code where, due to the permutation symmetry of the code, the conditional channel just depends on the weight of the correction operator and not on the detailed syndrome for uniform noise. 
Clearly symmetries of the error correcting code can simplify the effective channels that arise for uniform noise. 
Here we briefly discuss these simplifications. 
Symmetries of stabilizer codes have been studied previously, see for example~\cite{calderbank1998quantum, rains1999quantum, grassl2013leveraging, zeng2011transversality}. 

We will largely focus on the automorphism group of the code under permuting qubits, but we can formulate the notion of symmetry more generally. 
For our purposes, a stabilizer code together with its error correction procedure will have interesting symmetries when there is some subgroup $\mathsf{G}$ of the Clifford group acting on the physical qubits that preserves the error correction procedure. 
We will say that an error correcting code is symmetric under $\mathsf{G}$ if the logical operators can be chosen to be invariant under $\mathsf{G}$ and the stabilizer group is preserved under the action of $\mathsf{G}$. 
Specifically for all $G\in \mathsf{G}$ and stabilizer elements $\vec{t}$ there exists some $\vec{t}'$ such that  $G S_{\vec{t}}G^\dagger=S_{\vec{t}'}$, while for all $G\in \mathsf{G}$ and $\vec{l}$ we have $G L_\vec{l}G^\dagger =L_\vec{l}$. 
We will say that the decoding procedure is symmetric if the code is symmetric and for all $G$ and $\vec{t}$ there is some $\vec{t}'$ such that $G E_{\vec{t}}G^\dagger=E_{\vec{t}'}$. 
Up until now we have used the term symmetric decoder if the correction $E_\vec{t}^\dagger$ has the lowest weight of any error that could have caused the syndrome $\vec{t}$. 
We now place this further requirement on the choice of correction. 

When the error correction procedure is symmetric the syndromes will break up into orbits under the action of $\mathsf{G}$. 
Rather than considering general Clifford symmetries of stabilizer codes we will now specialize to the case where the symmetry of the code is some group of permutations of physical qubits. 
Qubit permutations are very special because they preserve the Hamming weight and $\alpha^\vec{n}$. 
Thus if we perform a permutation $G P_{\vec{b}}G^\dagger =P_{\vec{b}'}$ then $w(\vec{b}')=w(\vec{b})$ and so on. 
We will see that conditional channels of all syndromes in a given orbit are equal for such permutation symmetries of stabilizer codes. 
This explains the simplifications observed in \cref{sec:repetition} for the repetition code and in~\cite{Chamberland2017} for a range of other codes.

To demonstrate that the effective channels for syndromes in a single orbit are equal, consider a fixed syndrome $\vec{s}$ and a fixed permutation $G \in \mathsf{G}$. 
Then we define $\vec{s}'$ such that $G E_{\vec{s}}G^\dagger=E_{\vec{s}'}$. 
Now let us consider a single term in \cref{eq:fullsyndrome} that has the chosen $\vec{s}$, we have 
\begin{equation}
	P_{\vec{b}}= \eta_{\vec{s},\vec{l},\vec{t}}E_{\vec{s}}L_{\vec{l}}S_{\vec{t}}
\end{equation}
and if we define $\vec{b}'$ such that $GP_{\vec{b}}G^\dagger=P_{\vec{b}'}$ we have
\begin{equation}
P_{\vec{b}'}= \eta_{\vec{s},\vec{l},\vec{t}}(GE_{\vec{s}}G^\dagger)(G L_{\vec{l}}G^\dagger)(G S_{\vec{t}}G^\dagger) = \eta_{\vec{s},\vec{l},\vec{t}} E_{\vec{s}'} L_{\vec{l}} S_{\vec{t}'} = \eta_{\vec{s}',\vec{l},\vec{t}'} E_{\vec{s}'} L_{\vec{l}} S_{\vec{t}'} .
\end{equation}
The second \acdedit{equality} just arises from the definitions and the fact that the error correction procedure is symmetric under $\mathsf{G}$. 
The third equality holds because $\eta$ can be written in terms of the commutators of the various factors, and these are preserved under the action of the unitary. 
(Alternatively it can be written in terms of the symplectic inner product which is manifestly invariant under swapping qubits; see the appendix for more details.) Moreover we have
\begin{equation}
	w(\vec{s},\vec{l},\vec{t})= w(\vec{b})=w(\vec{b}')=w(\vec{s}',\vec{l},\vec{t}')
\end{equation}
since permutations leave the weight unchanged, and the same identity holds for $\alpha^\vec{n}$. 
The mapping takes $S_{\vec{t}}$ to $S_{\vec{t'}}$ under the group $\mathsf{G}$ so in the formula for the Kraus operator of the conditional channel (\ref{eq:generalKraus}) each term in the sum for $K_\vec{s}$ maps invertibly to a single term in the sum for $K_{\vec{s}'}$. 
Consequently  we have
\begin{eqnarray}
K_{\vec{s}}&=&  \sum_{\vec{l}}\sum_{\vec{t}} f_{w(\vec{s},\vec{l},\vec{t})}(\theta)\alpha^{\vec{n}(\vec{s},\vec{l},\vec{t})}\eta_{\vec{s},\vec{l},\vec{t}}L_{\vec{l}}\\
&=&  \sum_{\vec{l}}\sum_{\vec{t}'} f_{w(\vec{s}',\vec{l},\vec{t}')}(\theta)\alpha^{\vec{n}(\vec{s}',\vec{l},\vec{t}')}\eta_{\vec{s}',\vec{l},\vec{t}'}L_{\vec{l}} =K_{\vec{s}'}.
\end{eqnarray}
Thus all the Kraus operators corresponding to syndromes in a given orbit of the symmetry are equal.

We have shown these simplifications only for unitary noise at the physical level. 
It would be interesting to extend this analysis to noise processes with more than a single Kraus \acdedit{operator}. 

In the case of repetition codes with an odd number of qubits, we can choose the code and decoding procedure to have the full permutation symmetry of the $d$ qubits. 
This implies that the Kraus operators can only depend on the Hamming weight of the correction procedure, as we observed in \cref{sec:repetition}. 

In the case of the five-qubit code~\cite{Laflamme1996}, the stabilizer group is generated by the operators $XZZXI$ and its cyclic permutations and the logical operators can be chosen to be $\bar{X}=XXXXX$ and $\bar{Z}=ZZZZZ$. 
It is clear that this code is symmetric under cyclic permutations of the qubits. 
It is also symmetric under the permutation $(1, 5)(2,4)$, which is a reflection if the five qubits are arranged on the vertices of a pentagon. 
This generates the group $D_{10}$ of symmetries of the pentagon. 
(We believe that this is the full symmetry group of the code.)
The five qubit code is a $[5,1,3]$ code and there are $2^{n-k}=16$ distinct syndromes. 
In the symmetric decoder we choose correction operations that are the lowest weight Paulis that can result in each syndrome. 
These are the identity $IIIII$, and the fifteen single qubit Pauli's. 
This set of correction operators is symmetric under the full permutation group, so we can conclude that the overall error correction procedure is symmetric under $D_{10}$. 
The correction operations break into 4 orbits under this group. 
Representatives of these orbits are $IIIII,XIIII,YIIII$ and $ZIIII$. 
Consequently there are 4 distinct conditional effective channels for the 5 qubit code under the symmetric decoder. 
This agrees with the number of conditional effective channels in~\cite{Chamberland2017}.

For the Steane code~\cite{Steane1996} we can use similar arguments. 
This is a CSS code with its $X$-type stabilizers generated by $XXXIXII, XIXXIXI$ and $XXIXIIX$. 
The $Z$-type stabilizer generators are obtained by replacing $X$ with $Z$. 
The logical operators can be chosen to be $\bar{X}=XXXXXXX$ and $\bar{Z}=ZZZZZZZ$. 
The code is symmetric under the permutations $(2, 3)(6, 7)$, $(2, 3, 4)(5, 6, 7)$, $(1, 3)(2, 5)$, and $(1, 3)(4, 6)$. 
These permutations generate the 168-element group $\text{GL}(3,2)$, which is the group of symmetries of the Fano plane. 
One can check by exhaustive search that this is the full symmetry group of the code. 
In order to identify the symmetry group we found it useful to picture the Steane code as the smallest possible color code. 
From this point of view the qubits in the code correspond one-to-one with the points of the Fano plane. 
Moreover the lines of the Fano plane correspond to the logical operators with weight 3. 
(It would be interesting to know whether some relationship of this kind holds between other color codes and other finite geometries.) 
Since the Steane code is a $[7,1,3]$ code there are  $2^{n-k}=64$ distinct syndromes. 
In the symmetric decoder the correction operations are the identity $IIIIIII$, the 21 single qubit Paulis, and the 42 Paulis that are the tensor product of a single $X$ with a single $Z$. 
This set of correction operators is symmetric under the full permutation group, so we can conclude that the overall error correction procedure is symmetric under $\text{GL}(3,2)$. 
The correction operations break into 5 orbits under this group. 
Representatives of these orbits are $IIIIIII, XIIIIII, YIIIIII, ZIIIIII$, and $XZIIIII$. 
This number of orbits does not agree with the seven inequivalent conditional effective channels stated in~\cite{Chamberland2017} for general noise processes. 
However our result only applies to unitary channels (which may have higher symmetries than the more general channels considered in~\cite{Chamberland2017}) \acdedit{ and it is not clear whether the decoder in~\cite{Chamberland2017} has the highest possible symmetry. So} we expect that these statements are not in disagreement. 

Note that these simplifications are independent of the exact unitary error that has occurred and for certain highly symmetric noise processes it is possible that there are further reductions in the number of distinct conditional effective channels. 
We will see that an example of this is the 5 qubit code with unitary errors of the form $\exp[-i\theta(X+Y+Z)/\sqrt{3}]$.

For specific examples that we considered there are still patterns in the expressions for the conditional effective channels that arise from the structure of the code. 
One instance where such examples arise is when the code has a non-trivial single-qubit transversal gate.
Such a transversal gate maps the stabilizer group to itself but acts non-trivially on the logical operations, the corrections $E_\vec{s}$, and the factors $\alpha^\vec{n}$. 
In the example of the 5 qubit code the Clifford gate that maps $X\rightarrow Y\rightarrow Z\rightarrow X$ is transversal. 
This makes it possible to infer the conditional effective channel for an error in the orbit given by $YIIII$ in terms of the one for errors in the orbit containing $XIIII$. 

\section{Performance of stabilizer codes under coherent errors}\label{sec:unitarycodeperformance}

In this section we will review the performance of a variety of quantum error correcting codes against coherent errors using the techniques developed so far. 
Recall that for a unitary error we have $D_{\diamond}\approx 3\sqrt{r}/2 \approx |\theta|$ so that there is a large difference between the size of errors as estimated in randomized benchmarking and errors as measured by diamond norm distance. 
Characteristically we find that after error correction the diamond norm error of the effective channel $D'_\diamond$ is much smaller than would be expected based on the diamond norm error at the physical level. 
For dephasing errors $D'_\diamond \propto D_\diamond^{t+1}$ for a small error. 
On the other hand we have seen that for a unitary error $D'_\diamond\leq c D_\diamond^d$.  

The difference in this behavior is almost sufficient to overcome the distinction between diamond norm error and average gate infidelity at the physical level. 
If $r$ is the average fidelity error at the physical level then for an odd distance code we find $D'_\diamond\leq c(3/2)^{2t+1} r^{t+1/2}$ for unitary errors, while for dephasing errors $D'_\diamond \propto r^{t+1}$. 
For even distance codes this comparison is even more striking, we find $D'_\diamond\leq c(3/2)^{2t+2} r^{t+1}$ for unitary errors, compared to $D'_\diamond \propto r^{t+1}$ for dephasing errors. 
So the scaling with $r$ is the same for both unitary and dephasing errors. 
We emphasize that these results assume ideal error correction, and do not analyze fault tolerant gadgets and noisy gates and measurements. 
They do suggest that there is a need for a sharper analysis of the performance of fault tolerant gadgets against unitary errors. 

The remainder of this section explores the logical effective channels and the logical diamond norm error scaling relative to the physical diamond norm error and the physical average gate infidelity for a variety of small quantum codes using pure unitary errors around arbitrary axes. 
In the next section, we will also consider the more general noise model that includes dephasing and demonstrate with the Steane code how the computation of the effective logical channel can be used to find thresholds for concatenated codes with the symmetric decoder. 

\subsection{Repetition codes}

In the limit of low error $x,y\ll 1$ we can compute simple approximations for the diamond norm error of the effective channel of the repetition codes discussed in \ref{sec:repetition}. 
For pure dephasing noise we have $y=0$ and $D'_\diamond \simeq \binom{n}{t+1}x^{t+1} =  \binom{n}{t+1}D_\diamond^{t+1}=\binom{n}{t+1}(3/2)^{t+1}r^{t+1}$ which is the expected scaling with error probability, physical diamond norm error, and average gate infidelity respectively. 

For unitary errors we have $y^2=x(1-x)$ and it is possible to see that $D'_\diamond$ is dominated by the contribution arising from $y'$. 
Consequently we have $D'_\diamond \simeq \binom{2t}{t}|y|^{2t+1}\simeq \binom{2t}{t}D_\diamond^{2t+1}\simeq \binom{2t}{t} (3/2)^{t+1/2} r^{t+1/2}$, which is the scaling suggested by \cref{thm:theorem}. 
\acdedit{Notice that these results imply that a simple Pauli Twirling Approximation to a unitary error, which just sets $y=0$, would underestimate the logical error, finding $D'_{PTA\diamond} \simeq \binom{2t+1}{t+1}|\theta|^{2t+2}$ in the limit of small $|\theta|$, as opposed to $\binom{2t}{t}|\theta|^{2t+1}$.  In the terminology of \cite{magesan2013modeling} this implies that the Pauli Twirling Approximation to the effective noise channel is \emph{dishonest} but it is less dishonest than when the PTA is applied at the logical level. }

\subsection{Effective channels for unitary errors: the five-qubit code and others}

As discussed in \cref{sec:symmetric} the five qubit code is symmetric under cyclic permutation of the qubits of the code, and also possesses a symmetry resulting from the transversal Clifford gate operation that maps $X\rightarrow Y\rightarrow Z\rightarrow X$. 
Consequently we begin by considering unitary errors of the form $\exp[-i\theta (X+Y+Z)/\sqrt{3}]$ which enhances the symmetry of the various effective error channels. 
The five qubit code has 16 syndromes, and in this highly symmetric case there are only two conditional effective channels, one corresponding to the trivial syndrome. 
All the non-trivial syndromes result in the same effective channel. 
These two channels are each described by a single Kraus operator $K_0$ in the trivial case and $K_1$ in the non-trivial case. 
They can be written as follows
\begin{align}
    K_0 &= (g_0+15g_4)\bar{I} - (10g_3 - 6g_5)(\bar{X}+\bar{Y}+\bar{Z})\\
    K_1 &= (g_1+4g_3+3g_5)\bar{I} - (2g_2+2g_4)(\bar{X}+\bar{Y}+\bar{Z}),
\end{align}
where $g_w=f_w/(\sqrt{3})^w$. 

Inspecting the expressions for $K_0$ and $K_1$ we can see that this behavior is qualitatively very similar to what we observed for the repetition codes. 
It is straightforward to check that both $K_0$ and $K_1$ are proportional to a unitary and correspond to rotation about the same axis $(X+Y+Z)/\sqrt{3}$ as occurs at the physical level. 
In the case of a small rotation error $\theta\ll 2\pi$ we find that the trivial syndrome occurs with high probability and results in a rotation $\theta'_0\simeq 10 \theta^3/3$. 
A non-trivial syndrome occurs with probability $\simeq 15\theta^2$ and results in a rotation $\theta'_1\simeq -2\theta$. 
Notice that conditioned on a non-trivial syndrome the rotation angle has actually increased. 
Nevertheless this pattern is exactly the one seen in the repetition code, and saturates the scaling limit of \cref{thm:theorem} since we have $D'_\diamond \simeq c D_\diamond^3$ for small $\theta$. 
The qualitative behaviors that we observed for the repetition codes and special rotation axes are also seen in proper stabilizer codes for typical rotation axes. 

Consider now a rotation about the $Z$-axis. 
In this case we find four inequivalent Kraus operators. 
The first corresponds to the trivial syndrome. 
There are then five syndromes that detect single $Z$ errors and result in a effective error that is rotation about the $Z$ axis. 
The other two classes of syndromes would generally detect single $Y$ and $X$ errors. 
But since these do not arise for this noise model, they detect two-qubit $Z$ errors of the form $ZZIII$ and $ZIZII$ respectively. 
The resulting Kraus operators are
\begin{align}
    K_0 &= f_0 \bar{I} + f_5 \bar{Z}\\
    K_1 &= f_1\bar{I} +f_4\bar{Z} \\
    K_2 &= -f_2\bar{Y}-if_3 \bar{X} \\
    K_3 &= -f_2\bar{X}+if_3\bar{Y}.
\end{align}
Here $K_1$ is the conditional channel that results when the correction operation is a $Z$ operation, $K_2$ corresponds to a $Y$ operation and $K_3$ is an $X$ operation. 
The Kraus operators are proportional to unitaries, but this time the rotation axis is not necessarily the same as the original unitary. 
Notice that from these conditional effective channels we find $D'_\diamond \simeq c D_\diamond^4$ so that unitary errors about this axis are much better corrected than for a typical rotation axis. 

We see that for this Pauli axis rotation, the 5 qubit code behaves as if it corrects two errors rather than one, and if we modified the correction operation to correct for two-qubit $Z$ errors it would be possible to correct the unitary rotation such that $D'_\diamond \simeq c D_\diamond^5$. 
This is consistent with the fact that the five qubit code has distance 5 if there are only $Z$ errors, as it becomes a repetition code in this limit. 
Note that this improvement requires that the rotation axis of the unitary is known.

We will not write the effective channels for a general rotation axis, but it is illuminating to see how the performance of the code varies as the rotation axis is changed. 
This is shown in \cref{fig:5qubit}, where $D_{\diamond}'/D_{\diamond}^d$ is plotted for various rotation axes. 
The improved performance for rotations about the Pauli axes is clearly visible. 
It can be seen that rotations about axes like $(X+Y+Z)/\sqrt{3}$ are local maxima for the diamond norm error after error correction.

\begin{figure}[t!]
(a) \subfigure{
\includegraphics[width=.55\linewidth]{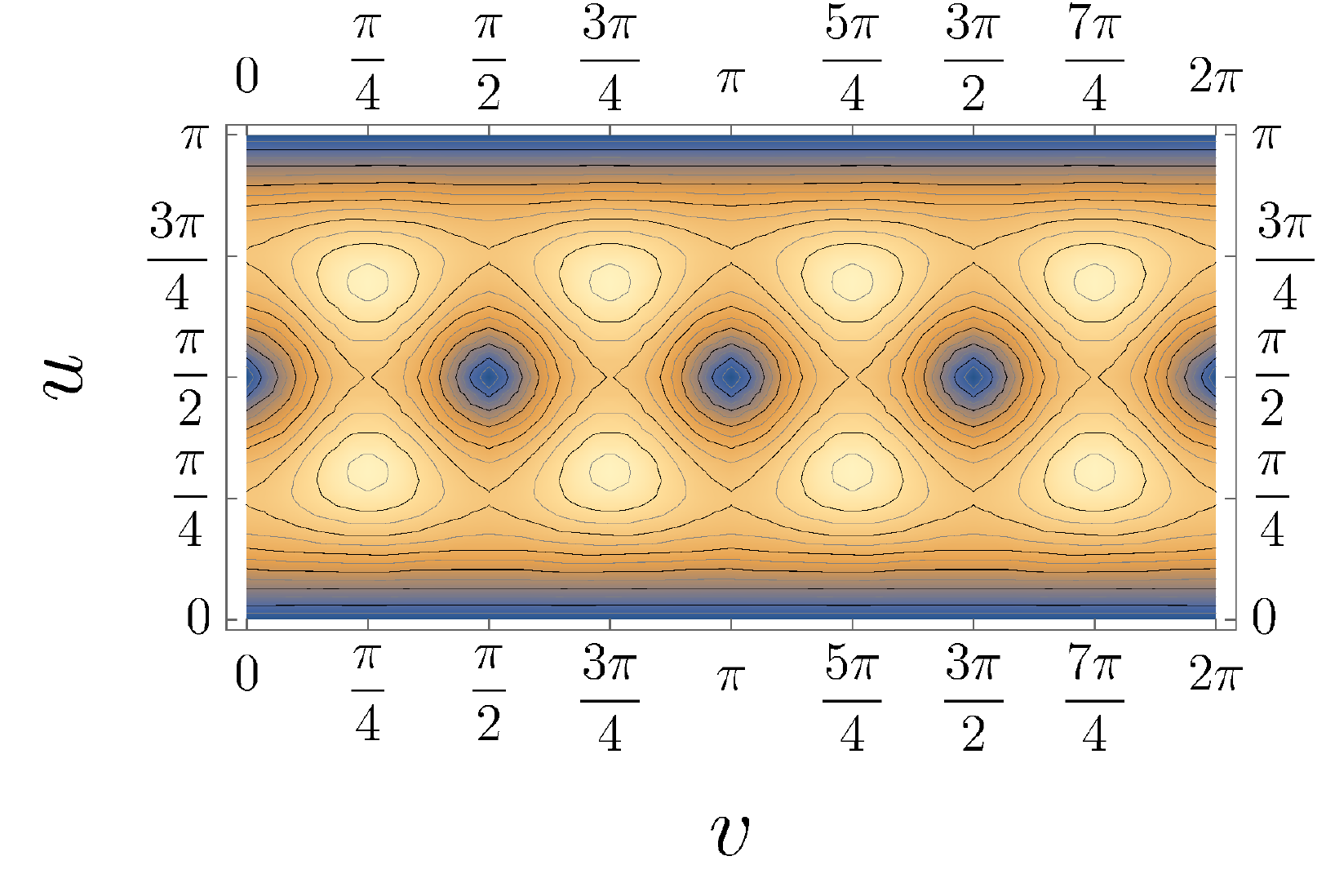}
\raisebox{0.17\height}{\includegraphics[height=0.05\linewidth,angle=90]{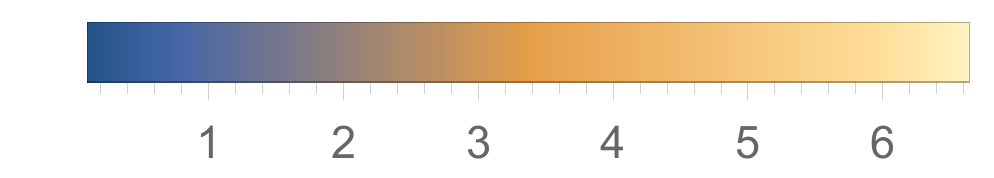}}
}
 (b) \subfigure{\includegraphics[width=0.3\linewidth]{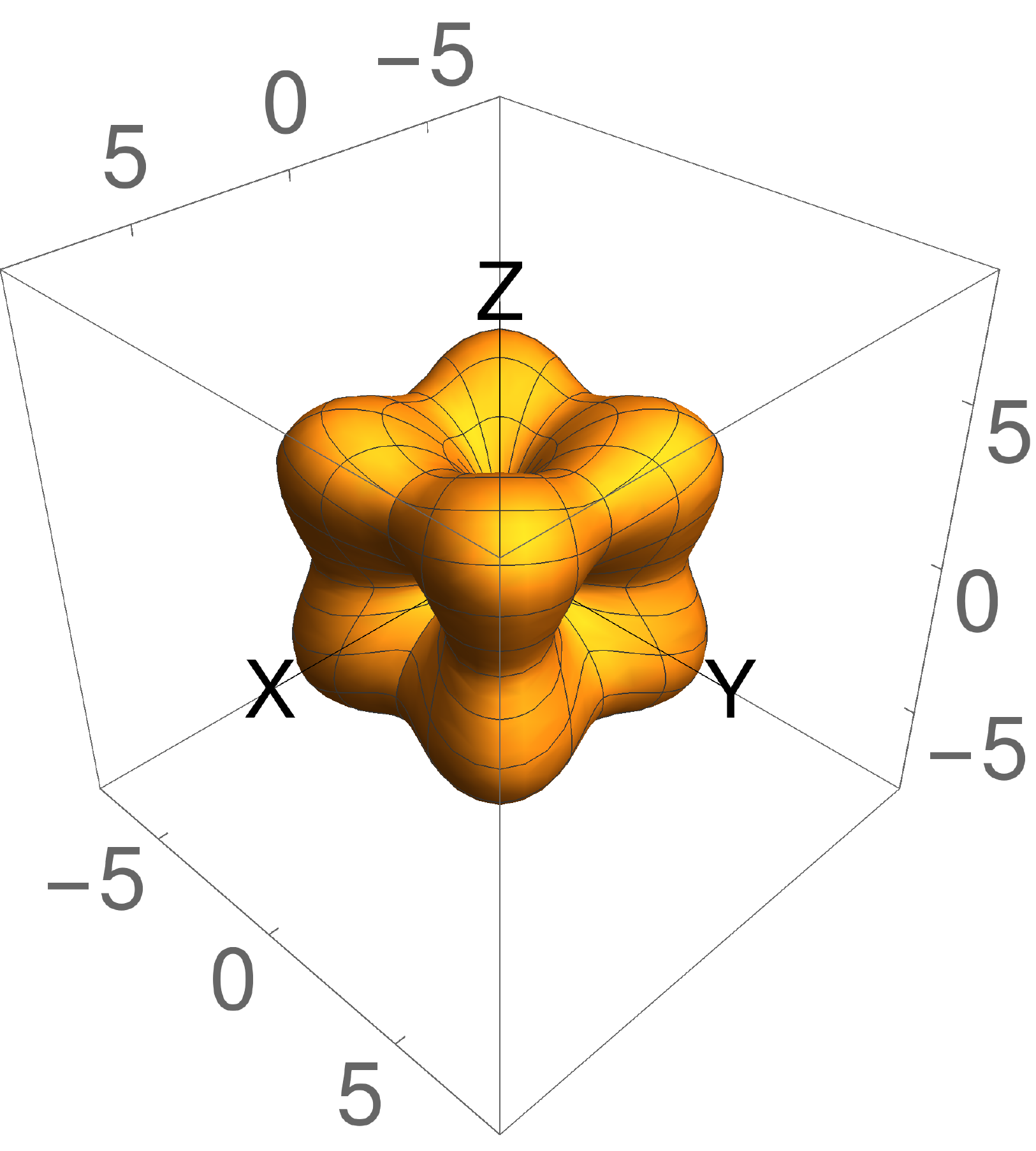}
}
\caption{(a) Contour plot of the ratio $D_{\diamond}'/D_{\diamond}^d$ against the rotation axis of a single qubit unitary error for the $[5,1,3]$ five-qubit code using the symmetric decoder. 
The polar angle is $u$ and the azimuthal angle is $v$. 
The rotation angle for the unitary error was chosen to be $\theta=0.01$. (b) 
3D spherical plot of the ratio $D_{\diamond}'/D_{\diamond}^3$ for the $[5,1,3]$ five-qubit code calculated as for (a). 
Rotations about Pauli axes are suppressed at a higher order than all other axes.}
\label{fig:5qubit}
\end{figure}

We can study the performance of many other stabilizer codes in the same way. 
The general behavior is similar to the five qubit code, although for typical codes the scaling of $D_{\diamond}'$ with $D_{\diamond}^d$ holds for all rotation axes. 
This is true of the Steane code, for example, as shown in \cref{subfig:Steane}, where we see improved performance for rotations about $(X+Y+Z)/\sqrt{3}$ and reduced performance for rotations about the Pauli axes. 

We provide a few further examples of the dependence of the effective channel on the rotation axis of the unitary which arises when using the symmetric decoder. 
In \cref{subfig:Li} we show the analogous plot for the Bare $[7,1,3]$ code of~\citet{Li2017} and in \cref{subfig:Shor} the $[9,1,3]$ Shor code~\cite{Shor1995}. 
In the case of the Shor code errors for rotations about $Z$ are corrected less well than rotations about $X$ due to the structure of this code as a concatenation of codes correcting bit flips and phase flips. \acdedit{This concatenation structure means that the subgroup of the stabiliser group made up of products of $X$ operators is larger than the subgroup made of products of $Z$'s.}
Lastly, we show the $[9,1,3]$ surface code~\cite{Bravyi1998} in \cref{subfig:surface}. 
Here again we see that there is enhanced correction along $(X+Z)/\sqrt{2}$ and especially along $Y$. 
This is consistent with the large increase in the error correction threshold observed for incoherent $Y$ noise in the surface code when using an optimal decoder~\cite{Tuckett2018}. 
\acdedit{Unfortunately, we do not have a general understanding of the symmetry axes that show enhanced performance for all of the codes that we plot in \cref{fig:balloon}, and it remains an intriguing open problem to better explain or predict such behavior in the most general case.}

\begin{figure}[t!]
(a) \subfigure{\includegraphics[width=0.21\linewidth]{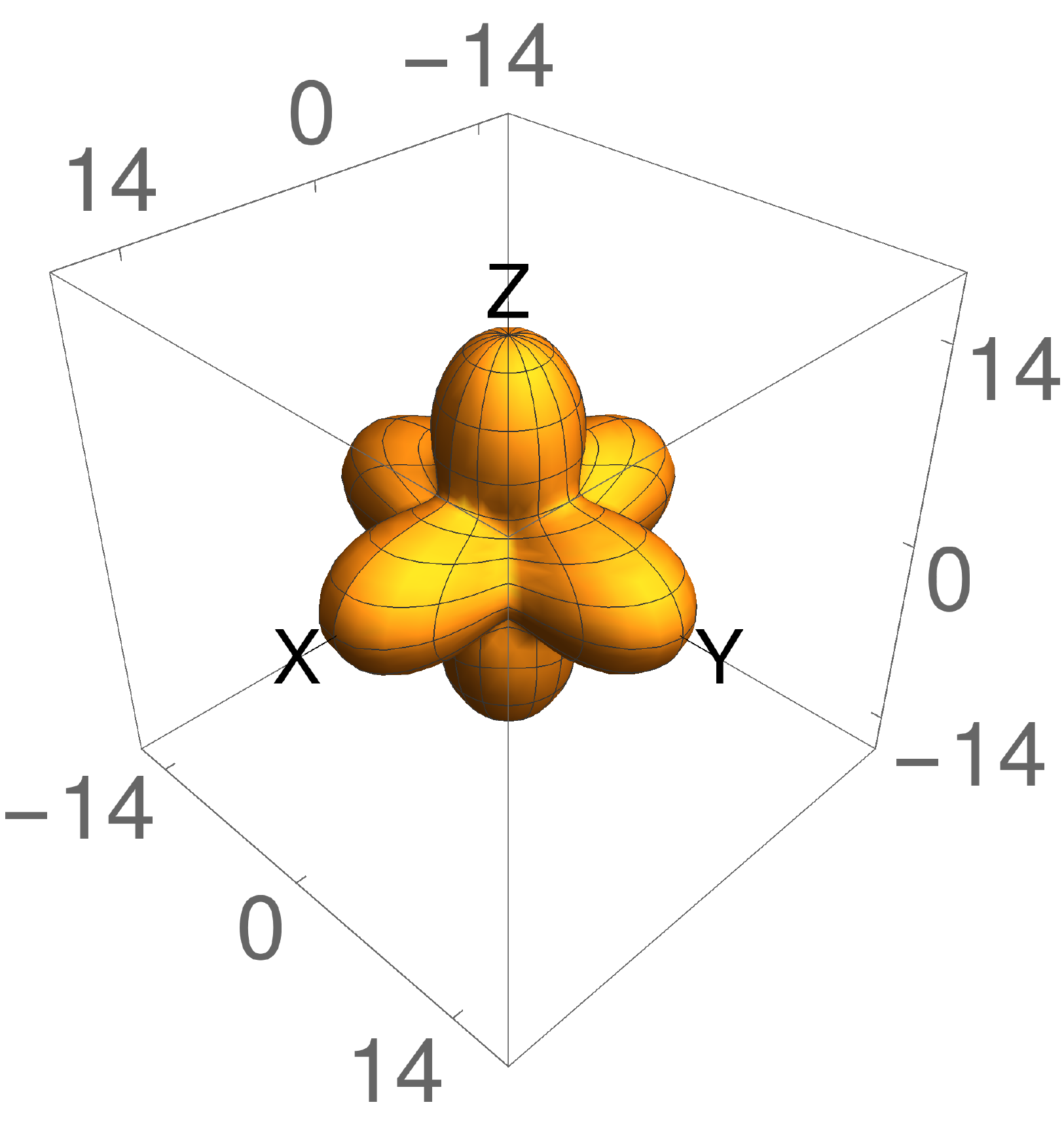}\label{subfig:Steane}}
(b) \subfigure{\includegraphics[width=0.21\linewidth]{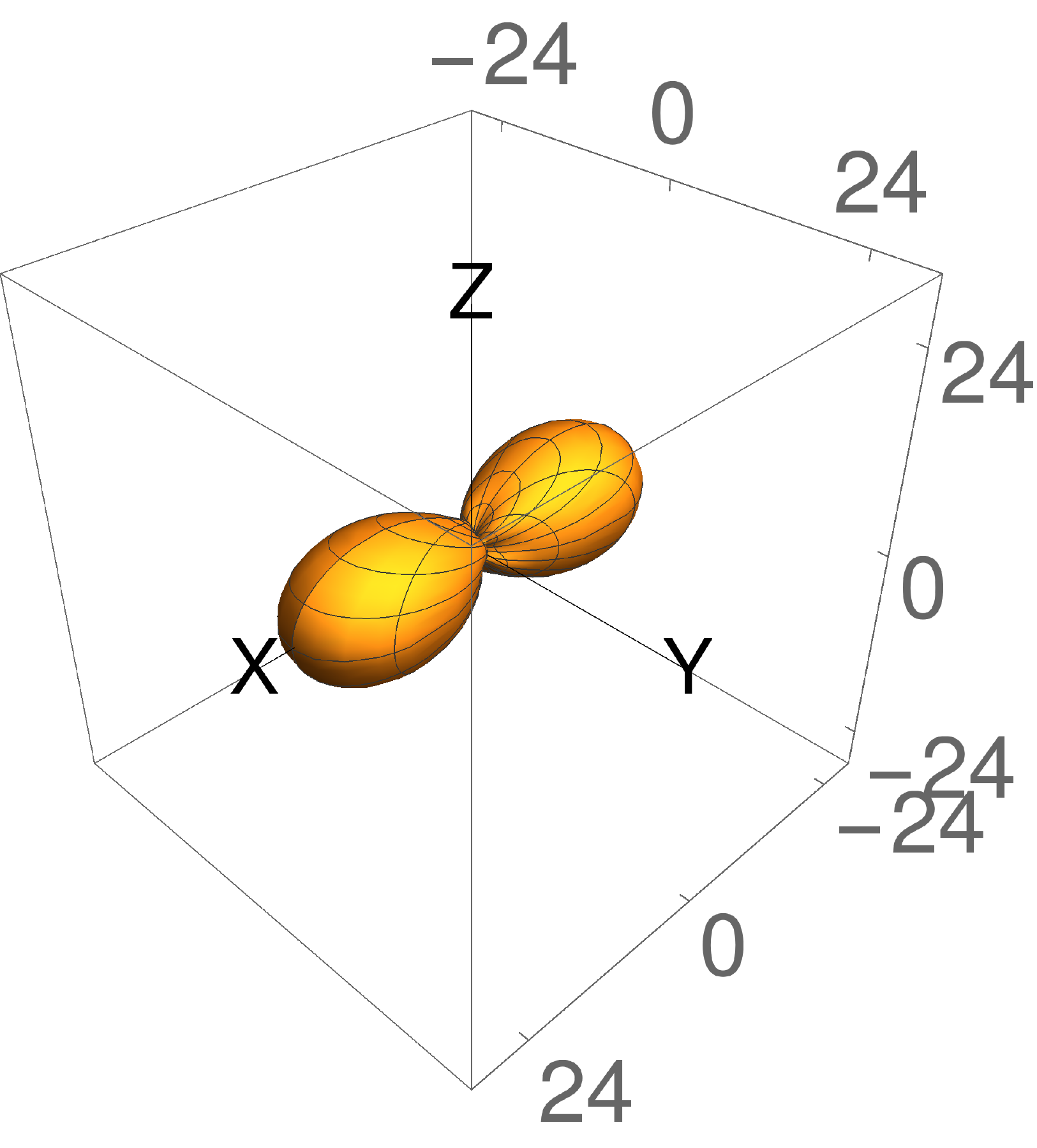}\label{subfig:Li}}
(c) \subfigure{\includegraphics[width=0.21\linewidth]{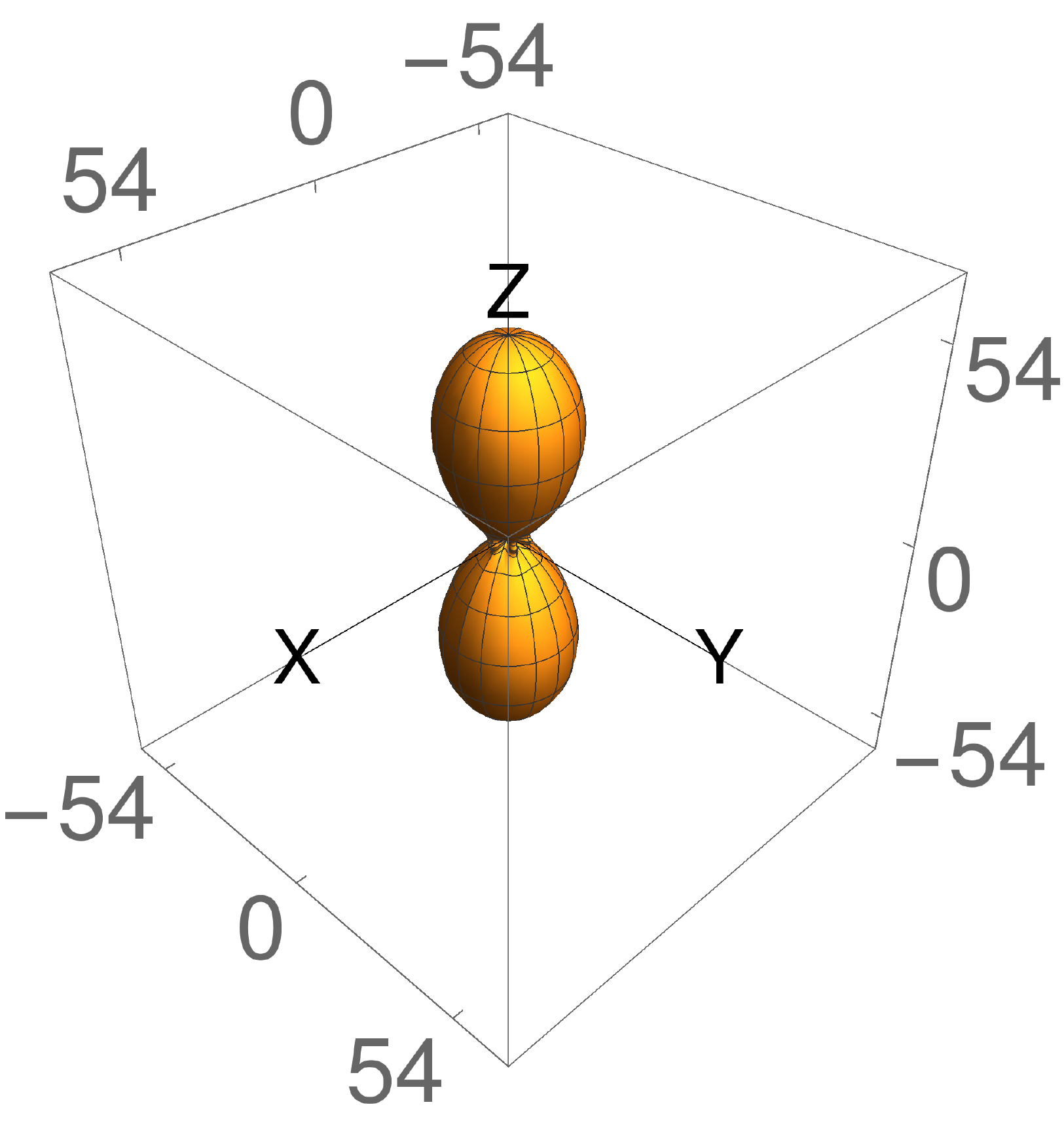}\label{subfig:Shor}}
(d) \subfigure{\includegraphics[width=0.21\linewidth]{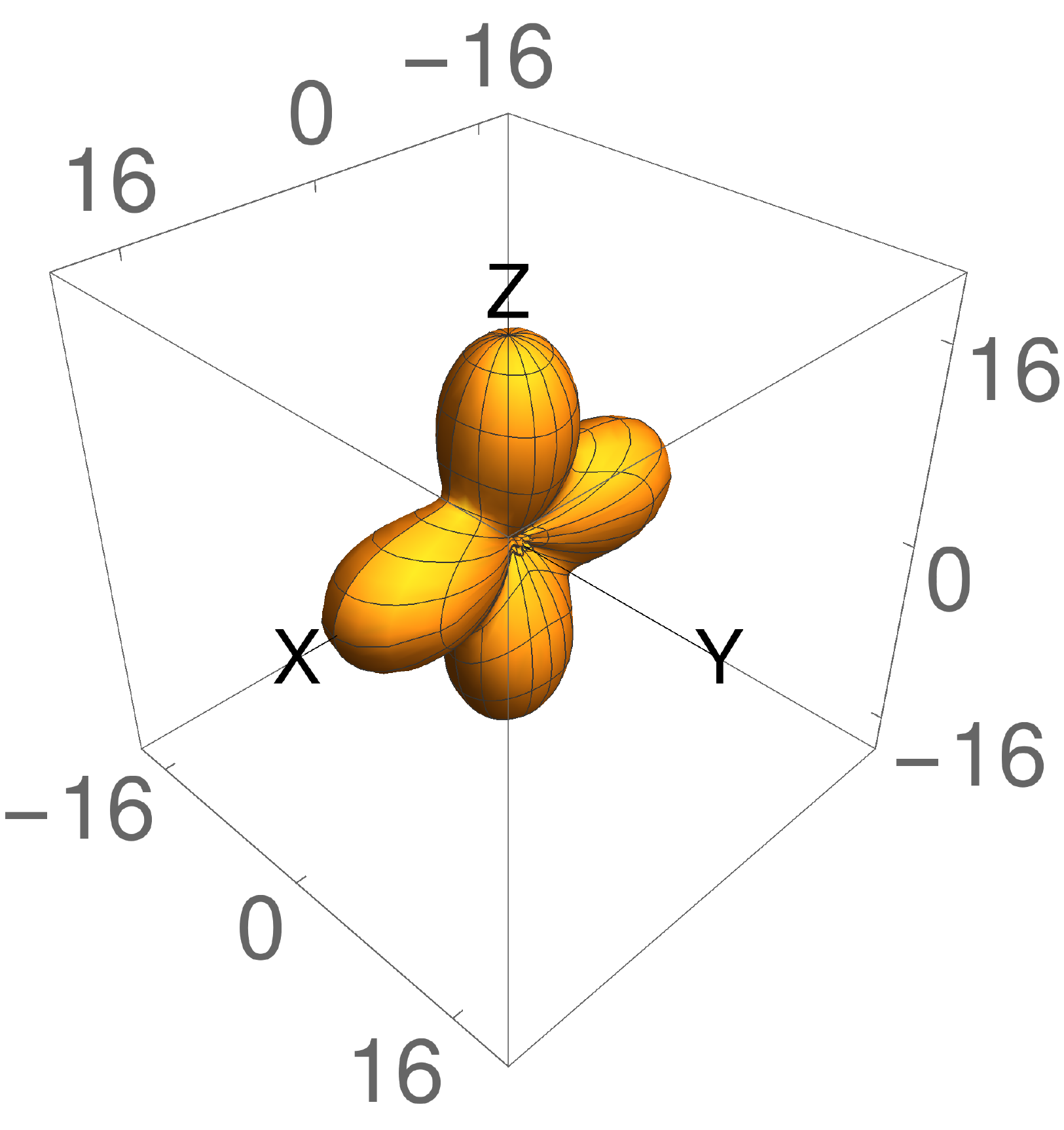}\label{subfig:surface}}
\caption{3D spherical plots of the ratio $D_{\diamond}'/D_{\diamond}^d$ against the rotation axis of a single qubit unitary error for various error correcting codes using the symmetric decoder. 
The rotation angle for the unitary error was chosen to be $\theta=0.01$. 
The codes are (a) the Steane code, (b) the Bare $[7,1,3]$ code~\cite{Li2017}, (c) the $[9,1,3]$ Shor code, and (d) the $[9,1,3]$ surface code.
Different codes have different axes along which the post-error correction diamond norm error is more favorable than others. }
\label{fig:balloon}
\end{figure}

\section{Concatenated code thresholds for joint unitary and dephasing errors}\label{sec:concatcodeperformance}

We now return to the noise model of \cref{eq:noisemodel}, which combines $Z$-axis rotation with dephasing. 
Whereas previously we have always used the symmetric, minimum distance decoder, we now specialize to a $Z$-only decoder that finds the minimum distance correction among all Paulis consisting only of $Z$ and $I$. 
This is sensible since for these noise channels only $Z$-type errors are supported. 

For any given code we can regard the ideal error correction as a map on the noise channel. 
For most of the codes we have investigated we find that the form of the noise is unchanged and there exist some $(x',y')$ such that the error correction implements the map $(x,y)\rightarrow (x',y')$. 
This type of calculation was also performed by \citet{Gutierrez2016}, but in contrast we don't need the full process matrix since our noise model is restricted to \cref{eq:noisemodel} and so is a function solely of $x$ and $y$. 
For example, for the Steane code we find 
\begin{align}
\label{eq:Steane}
    \begin{split}
        x' &= 21 x^2-98 x^3+210 x^4-252 x^5+168 x^6-48 x^7\\
        & +42 y^4-252 x y^4+504 x^2 y^4-336 x^3 y^4
    \end{split}\\
    y' &= 14 y^3-168 x y^3+504 x^2 y^3-672 x^3 y^3+336 x^4 y^3+48 y^7,
\end{align}
so to lowest order, under unitary noise only, $p' \approx 63\theta^4$ and $\theta'\approx 14\theta^3$, which is consistent with \cref{thm:theorem}. 

This procedure can handle rather larger codes. 
As a further example, consider the $[16,1,4]$ surface code described in~\cite{nickerson2016error}. 
We consider the noise model of \cref{eq:noisemodel} and perform correction by measuring $X$-stabilizers only. 
With the Z-only decoder, we find $y' = 0$ and 
\begin{align}
    \begin{split}
        x' ={}& 32 x^2 - 188 x^3 + 484 x^4 - 500 x^5 - 612 x^6 + 3136 x^7 - 5680 x^8 + 6080 x^9\\
        &- 4032 x^{10} + 1536 x^{11} - 256 x^{12} + 28 x y^2 - 222 x^2 y^2 + 756 x^3 y^2 - 1378 x^4 y^2\\
        &+ 1104 x^5 y^2 + 752 x^6 y^2 - 2880 x^7 y^2 + 3120 x^8 y^2 - 1600 x^9 y^2 + 320 x^{10} y^2 \\
        &+ 12 y^4 - 96 x y^4 + 208 x^2 y^4 + 416 x^3 y^4 - 3088 x^4 y^4 + 7040 x^5 y^4 - 8320 x^6 y^4\\
        &+ 5120 x^7 y^4 - 1280 x^8 y^4 - 10 y^6 - 32 x y^6 + 704 x^2 y^6 - 2880 x^3 y^6 + 5280 x^4 y^6\\
        &- 4608 x^5 y^6 + 1536 x^6 y^6 - 48 y^8 + 384 x y^8 - 1152 x^2 y^8 + 1536 x^3 y^8\\
        &- 768 x^4 y^8 - 16 y^{10} + 64 x y^{10} - 64 x^2 y^{10}
    \end{split}
\end{align}
It is striking that for this even distance error correcting code we reproduce the behavior for even distance repetition codes where we also found $y'=0$. 
This confirms that the behavior seen there also arises for other more interesting quantum error correcting codes, although we still do not know if this is true for every even distance code in general. 

As discussed by \citet{Rahn2002} the effective channel can be used to study the performance of concatenated codes. 
The effective channel can be regarded as a map on noise processes, and iterating this map corresponds to concatenated error correction. 
As always the noise map corresponds to a particular choice of decoder. 
In this case the corrections are made at each level of the code using only the syndrome at that level. 
The full syndrome of the concatenated error correcting code is \emph{not} used to determine the correction. 
Such decoders are known as \emph{hard decoders} and it has been shown that soft decoders, that do consider the full syndrome, can lead to improved performance~\cite{Poulin2006}. 
Moreover simply iterating the mapping means that at each level of the code the symmetric decoder is used. 
It has also been shown that there are hard decoders that can improve on this~\cite{Chamberland2017}. 
(The improvements in~\cite{Chamberland2017} typically depend on knowledge of the noise model, whereas the concatenated symmetric or $Z$-only decoders can be specified independently of the precise noise model.) 
Nevertheless it is of interest to study the threshold of the concatenated symmetric decoder under unitary noise. 
For example it was shown in~\cite{Greenbaum2016} that after two levels of concatenation through a repetition code a unitary error of the form $\exp(-i\theta Z)$ results in an effective noise process dominated by dephasing.

\begin{figure}[t!]
(a) \subfigure{\includegraphics[trim=15mm 0mm 0mm 0mm,clip,width=0.50\columnwidth]{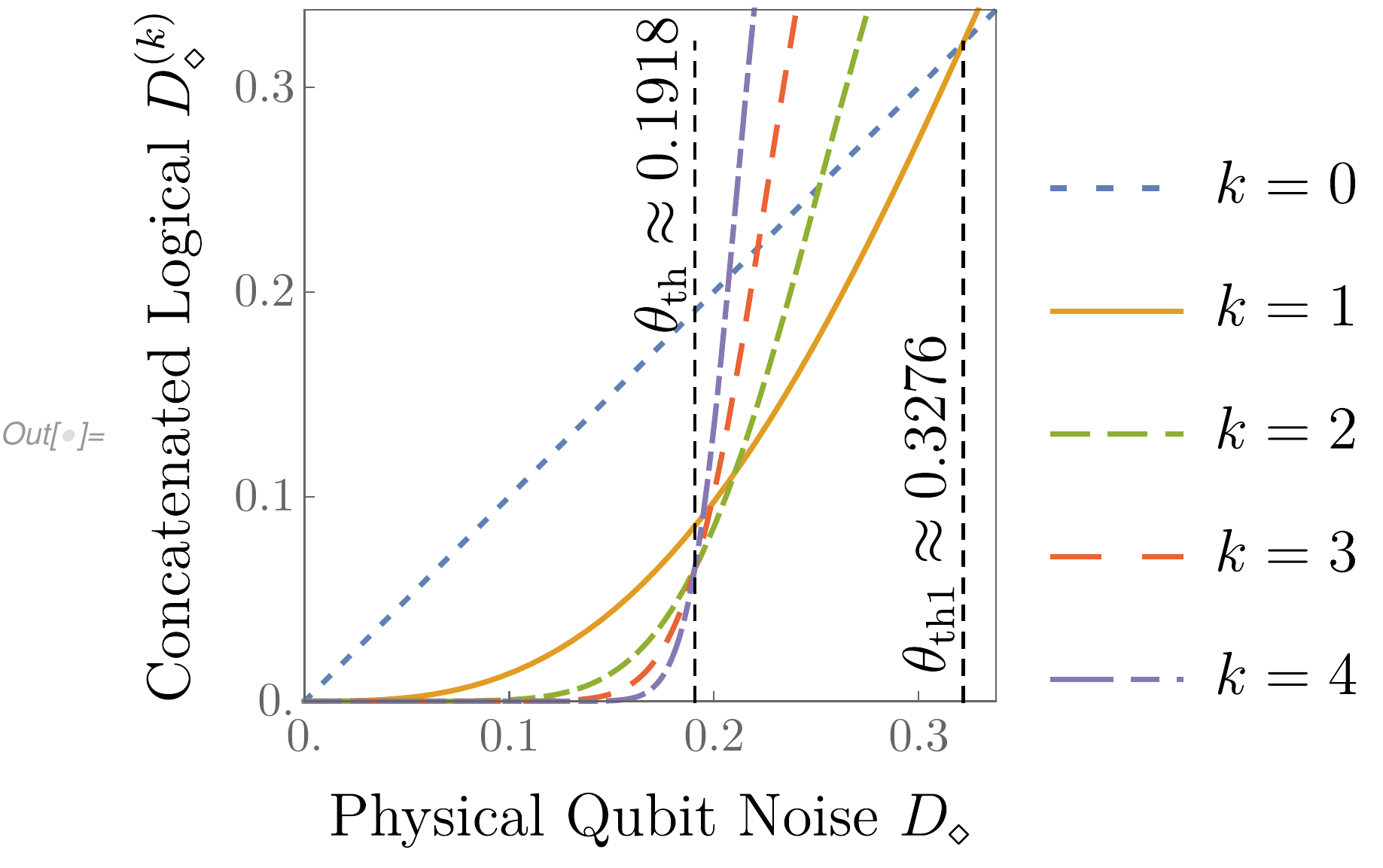}\label{subfig:SteaneCat}}
(b) \subfigure{\includegraphics[width=0.36\columnwidth]{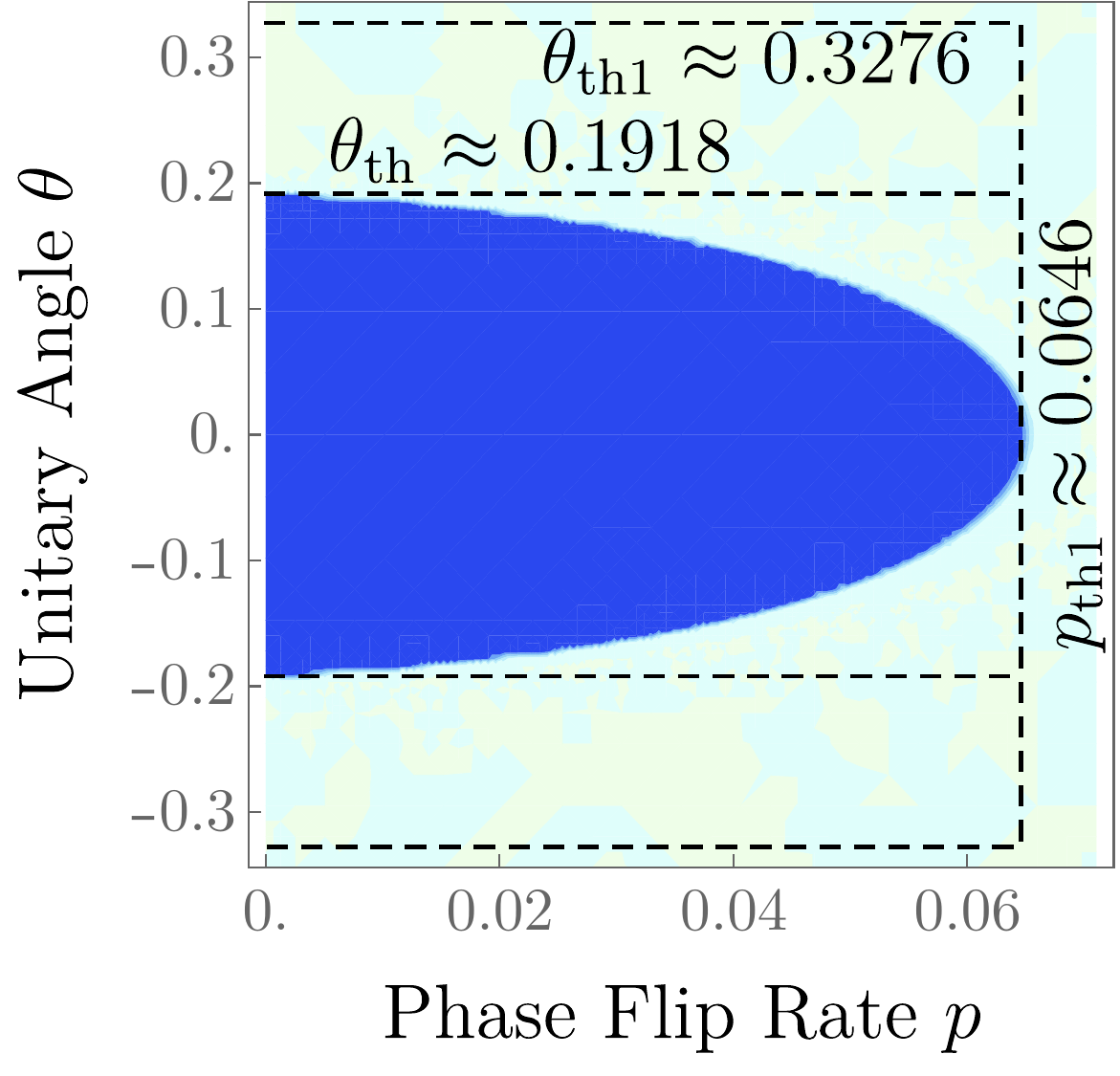}\label{subfig:SteaneAttractor}}
\caption{(a) Plots of the diamond distance $D_{\diamond}^{(k)}$ of the effective channel for the $k$-th concatenation of the Steane code over the diamond distance $D_{\diamond}$ of the original unitary $X$ noise channel for $k=0,1,2,3,4,5$. 
The true threshold and level-1 pseudothresholds are indicated by dashed lines. 
(b) Thresholds for concatenation of the Steane code under unitary bit flip noise as a basin of attraction in $p$ and $\theta$. 
Thresholds and level-1 pseudothresholds for coherent and incoherent noise are labeled with dashed lines.}
\end{figure}

Looking at the Steane code again, the mapping in \cref{eq:Steane} can be concatenated to further reduce the noise in terms of the diamond distance for unitary noise as in \cref{subfig:SteaneCat}. 
The level-1 pseudothreshold is 
\begin{align}
    \theta_{\text{th1}}\approx 0.3276 \qquad D_{\diamond\,\text{th1}} \approx 0.3218.
\end{align}
For comparison, under incoherent dephasing noise alone, the level-1 pseudothreshold for concatenation of the Steane code is 
\begin{align}
    p_{\text{th1}}\approx 0.0646 \qquad D_{\diamond\,\text{th1}}\approx 0.0646.
\end{align}
As can be seen from \cref{subfig:SteaneAttractor}, this is also equal to the true threshold. 
This is because the logical noise model is exactly of the same form at each level of concatenation. 
By contrast, the true threshold for unitary $Z$ rotation under concatenation is 
\begin{align}
    \theta_{\text{th}} \approx 0.1918 \qquad D_{\diamond\text{th}} \approx 0.1906,
\end{align}
which is significantly lower than the level-1 pseudothreshold. 
We can look at the behavior over the full range of parameters to find the region for which the noise is below threshold and can be reduced to an arbitrarily small value by concatenating the code. 
This is done by considering the iterative map as a dynamical system and finding the basin of attraction, as indicated in \cref{subfig:SteaneAttractor}. 

We can repeat these calculations for a range of codes, as before. 
The results are summarized in \cref{tab:summary}. 
It is worth noting that the true threshold for the two extremal noise models of pure unitary $Z$ rotation error or pure dephasing error exhibit a large gap in each of the cases we examine, and perhaps surprisingly the threshold for coherent noise is higher in every case.

\begin{table*}[t]
\centering
\begin{tabular}{rlcccccc@{}}
	\toprule
	\multicolumn{2}{c}{Code} & \multicolumn{2}{c}{Incoherent} & \multicolumn{4}{c}{Coherent} \\ 
		 \cmidrule(rl){1-2}  \cmidrule(rl){3-4} \cmidrule(rl){5-8}
	$[n,k,d]$ & name & $D_{\diamond}(p_{\text{th}})$ & $r(p_{\text{th}})$ & $D_{\diamond}(\theta_{\text{th}})$  & $r(\theta_{\text{th}})$ & $D_{\diamond}(\theta_{\text{th1}})$ & $r(\theta_{\text{th1}})$  \\ 
	\midrule
    $[5,1,3]$	& 5-qubit	& 0.5000	& 0.3333 & 0.7071	& 0.3333 & 0.8284 & 0.4575 \\
    $[7,1,3]$	& Steane	& 0.0646	& 0.0431 & 0.1906	& 0.0242 & 0.3218 & 0.0690 \\
    $[7,1,3]$	& Bare	& 0.5000	& 0.3333 & 0.7077	& 0.3333 & 0.8761 & 0.5117 \\
    $[9,1,3]$	& Shor	& 0.0499	& 0.0333 & 0.1177	& 0.0092 & 0.1477 & 0.0145 \\
    $[9,1,3]$	& Surface	& 0.0753	& 0.0502 & 0.2053	& 0.0281 & 0.3085 & 0.0634 \\
    $[16,1,4]$	& Surface	& 0.0395	& 0.0263 & 0.1613	& 0.0173 & 0.3180 & 0.0674 \\ 
	\bottomrule
\end{tabular}
\caption{Concatenated code thresholds in terms of physical diamond norm error and \acdedit{and average fidelity error} for unitary $Z$ noise and incoherent dephasing using minimum-weight $Z$ decoders. 
In each case, the threshold diamond norm error exhibits two important effects. 
First, the level-1 pseudothreshold differs by a significant amount from the true threshold only for the case of unitary noise. 
Second, the threshold is substantially higher in the case of unitary errors compared to incoherent errors. \acdedit{Notice however that when the thresholds are expressed in terms of the average fidelity error then the threshold for unitary error is less than or equal to the threshold for a dephasing error. }}
\label{tab:summary}
\end{table*}

\section{Conclusion}

We have studied the performance of quantum error correction against coherent noise, and coherent noise together with dephasing. 
We have proven that the logical effective channel corrects coherent errors much more effectively than incoherent errors, in the precise sense of \cref{thm:theorem}. 
We have also shown that a natural family of decoders exhibits this effect for a wide range of codes both in the single-level coding regime, and for concatenated codes. 
The thresholds that we calculate suggest that quantum error correction performs better \acdedit{than expected} when the errors are coherent, at least in the sense that this seems to lead to higher thresholds \acdedit{when these thresholds are expressed in terms of diamond norm errors}. 

There are many open questions suggested by this work, but the most interesting is to extend this analysis to the context of fault tolerant quantum computation. 
It would be extremely interesting if the increased thresholds that we observed for ideal error correction continued to hold for fault tolerant circuits. 
This would in turn strongly suggest that a sharper analysis of the threshold is needed to get accurate estimates of the threshold when coherent noise is taken into account. 
It would also be especially interesting to investigate what further improvements could be obtained for these channels using optimal Pauli recovery channels instead of minimum distance decoding, or other structured near-optimal recoveries~\cite{Fletcher2008}. 

\begin{acknowledgments}
This work was supported by the Australian Research Council through the Centre of Excellence in Engineered Quantum Systems (CE110001013 and CE170100009), and by the US Army Research Office grant numbers W911NF-14-1-0098 and W911NF-14-1-0103. 
\end{acknowledgments}

\section*{Appendix}

In this appendix we will collect some formulas from the stabilizer formalism that are standard (see for example \cite{Nielsen2000}) but the discussion here will fix notation that will be required in the main text. 

We will write every Pauli matrix on $n$ qubits in terms of two length $n$ binary vectors  $\vec{b}_x$ and $\vec{b}_z$. 
Ones in this vector will correspond to a Pauli matrix acting on the corresponding qubit. 
The {\it support} of the vector $\vec{b}_x$ will be the set of qubits for which the corresponding entries of $\vec{b}_x$ are equal to one. 
The size of the support will be $|\vec{b}_x|$. 
We will also define the binary vector $\vec{b}_x.\vec{b}_z$ which is the bitwise product of the two vectors, and has ones on the intersection of the supports of $\vec{b}_x$ and $\vec{b}_z$. 
These are the qubits on which the corresponding Pauli matrix acts with $Y$. 
Recalling that $XZ=-iY$ we have that a general Pauli matrix can be written
\begin{equation}
P_{\vec{b}}=i^{|\vec{b}_x.\vec{b}_z|}X^{\otimes \vec{b}_x}Z^{\otimes \vec{b}_z}
\end{equation}
and it is clear that the Pauli matrices are in one-to-one correspondence to the $4^n$ binary vectors $\vec{b}_x,\vec{b}_z$. 
Note that the Pauli matrices $P_{\vec{b}}$ so defined are the products of Pauli's with a co-efficient $+1$ and are therefore Hermitian. 
This Pauli has weight $w=|\vec{b}_x|+|\vec{b}_z|-|\vec{b}_x.\vec{b}_z|$, there are $m_x=|\vec{b}_x|-|\vec{b}_x.\vec{b}_z|$ $X$ operators, $m_z=|\vec{b}_z|-|\vec{b}_x.\vec{b}_z|$ $Z$ operators and $m_y=|\vec{b}_x.\vec{b}_z|$ $Y$ operators in this product, so that $m_x+m_y+m_z=w$.

We will imagine that $\vec{b}_x$ and $\vec{b}_z$ are row vectors, and so we can use the length $2n$ row vector $\vec{b}=(\vec{b}_x,\vec{b}_z)$ to denote a general Pauli matrix. 
We can define a $2n$-by-$2n$ symplectic matrix
\begin{equation}
\Lambda =
\left[
\begin{array}{cc}
 0 &  I    \\
- I &  0  
\end{array}
\right]
\end{equation}
We can define a symplectic inner product as follows $\vec{b}\Lambda\vec{b}'^T=|\vec{b}_x. \vec{b}'_z|-|\vec{b}_z\cdot\vec{b}'_x|$. 
The two elements $\vec{b}$ and $\vec{b}'$ commute if and only if $(-1)^{\vec{b}\Lambda\vec{b}'^T}=1$. 
While we have defined the symplectic inner product as a number, we will only ever use its value modulo 4. 
(Notice that there is a sign that is different compared to Nielsen and Chuang~\cite{Nielsen2000} and we don't evaluate inner products modulo 2, but rather use the conventional inner product.)

We can therefore compute products straightforwardly
\begin{eqnarray}
P_{\vec{b}}P_{\vec{b}'}&=&i^{|\vec{b}_x.\vec{b}_z|}i^{|\vec{b}'_x.\vec{b}'_z|}X^{\otimes \vec{b}_x}Z^{\otimes \vec{b}_z} X^{\otimes \vec{b}'_x}Z^{\otimes \vec{b}'_z}\\  
&=&(-1)^{|\vec{b}'_x. \vec{b}_z|} i^{|\vec{b}_x.\vec{b}_z|}i^{|\vec{b}'_x.\vec{b}'_z|}X^{\otimes \vec{b}_x}X^{\otimes \vec{b}'_x}Z^{\otimes \vec{b}_z} Z^{\otimes \vec{b}'_z} \\ 
&=&(-1)^{|\vec{b}'_x\cdot \vec{b}_z|}
(-i)^{|\vec{b}'_x.\vec{b}_z|}(-i)^{|\vec{b}_x.\vec{b}'_z|} P_{\vec{b}+\vec{b}'}
=  (-i)^{|\vec{b}_x.\vec{b}'_z|-|\vec{b}'_x.\vec{b}_z|} P_{\vec{b}+\vec{b}'} =
(-i)^{\vec{b}\Lambda\vec{b}'^T}P_{\vec{b}+\vec{b}'} \label{eq:pauliprod}
\end{eqnarray}
The second line makes use of the commutation formula we have just derived. 
In the final line we use $(\vec{b}_x+\vec{b}'_x).(\vec{b}_z+\vec{b}'_z)= \vec{b}_x.\vec{b}_z+\vec{b}'_x.\vec{b}'_z+\vec{b}'_x.\vec{b}_z+\vec{b}_x.\vec{b}'_z$. 
As claimed the symplectic inner product only needs to be evaluated modulo 4. 
Since two Paulis commute if and only if the symplectic inner product is $0\mod 2$, when we multiply two commuting Paulis the phase factor is $\pm1$, conversely when multiplying two anticommuting Paulis the phase factor is $\pm i$.

Consider now a stabilizer code. 
The stabilizer generators are signed products of Pauli matrices $\pm P_{\vec{g}_i}$  where the binary vectors $\vec{g}_i$ satisfy certain constraints, for example that the stabilizer generators all commute. 
Since the generators are associated with $n-k$ row vectors $\vec{g}_i$ we can collect these vectors into a $(n-k)$-by-$2n$ binary matrix $G$ whose rows are the $g_i$. 
For any given error $E=P_{\vec{b}}$ the \emph{syndrome} is a length $n-k$ binary column vector $\vec{s}$ given by 
\begin{equation}
\vec{s}(\vec{b}) = G\Lambda \vec{b}^T \ \mod 2
\end{equation}
A general member of the stabilizer group $S_{\vec{t}}$ is described by the length $n-k$ binary vector $\vec{t}$ such that
\begin{equation}
S_{\vec{t}} =\pm \prod_{i=1}^{n-k} P_{\vec{g}_i}^{t_i} = \pm P_{\vec{t}G}.
\end{equation}
The first factor of $\pm 1$ is there to account for the signs of the stabilizer generators in product. 
In examples of interest to us, this will always be one. 
The phase factor in the final expression arises from iterating equation (\ref{eq:pauliprod}) and is real because all the stabilizer generators commute. 
It has a definite value that is suppressed in the interests of lightening the notation. 

The code has $k$ logical $X$ operators $\bar{X}_i= P_{\vec{x}_i}$ and $k$ logical $Z$ operators $\bar{Z}_i= P_{\vec{z}_i}$. 
Since the X operators are specified by $k$ binary row vectors of length $2n$, we can specify the logical $X$ operators by the $k$-by-$2n$ binary matrix $X$ whose rows are the $\vec{x}_i$, and likewise a corresponding matrix $Z$ for the logical $Z$ matrices.

We can describe products of logical Pauli operators $L_{\vec{l}}$  by a length $2k$ binary vector $\vec{l}=(\vec{l}_x,\vec{l}_z)$, in analogy to the Pauli matrices as follows
\begin{equation}
L_{\vec{l}}=(-i)^{|\vec{l}_x.\vec{l}_z|}\prod\bar{X}_i^{l_{xi}}\prod
\bar{Z}_j^{l_{zj}} = \pm (-i)^{|\vec{l}_x.\vec{l}_z|} P_{\vec{l}_xX}P_{\vec{l}_zZ}
\end{equation}
The phase factor $\pm1$ arises because the product of logical X's is a product of commuting operators and likewise the product of logical Z's. 
Again, it has a definite value, given by iterating equation (\ref{eq:pauliprod}), but suppressing this factor lightens the notation. 
Notice that the matrix product in the subscript results in a length $2n$ binary row vector as expected.

There are $2^{n-k}$ syndromes, and with each syndrome $\vec{s}$ we will associate the lowest weight error $E_\vec{s}$ that leads to that syndrome. 
Our decoder will reverse this error. 
Explicitly then $E_{\vec{s}}=P_{\vec{\tilde{b}}}$ where $\vec{\tilde{b}}$ is the lowest weight solution to $\vec{s}=G\Lambda \vec{\tilde{b}}^T$. 
Then it can be shown that every Pauli matrix can be written in the form $P_\vec{b}=\eta_\vec{b}E_{\vec{s}(\vec{b})}L_{\vec{l}(\vec{b})}S_{\vec{t}(\vec{ b})}$ where $\eta_\vec{b}$ is a phase factor equal to $\pm1$ or $\pm i$. 

\bibliography{refs}

\end{document}